\documentclass[10pt,journal]{IEEEtran}
\usepackage{amsmath}
\usepackage{amsthm}
\usepackage{amsfonts}
\usepackage{epsfig}
\usepackage{macros}
\usepackage{caption}
\usepackage{subcaption}
\usepackage{MnSymbol}
\usepackage{graphicx}
\usepackage{xcolor}
\usepackage{color}
\usepackage{tikz}
\usepackage{cite}
\usepackage{mathtools}
\usepackage{epstopdf}
\DeclareGraphicsExtensions{.pdf}
\newtheorem{definition}{Definition}
\newtheorem{theorem}{Theorem}
\newtheorem{lemma}{Lemma}
\usepackage{tikz}
\title{\LARGE \bf
{Almost-global tracking of the unactuated joint in a pendubot
}}
\author{Aradhana Nayak$^{1}$  and Ravi N. Banavar$^{2}$
\thanks{$^{1}$ Aradhana Nayak and $^{2}$Ravi N. Banavar are with the Systems and Control Enginerring,
Indian Institute of Technology Bombay, Mumbai, Maharashtra 400076, India
        {\tt\small aradhana@sc.iitb.ac.in, banavar@iitb.ac.in}}%
}

\begin{document}

\maketitle
\thispagestyle{empty}
\pagestyle{empty}

%%%%%%%%%%%%%%%%%%%%%%%%%%%%%%%%%%%%%%%%%%%%%%%%%%%%%%%%%%%%%%%%%%%%%%%%%%%%%%%%

\begin{abstract}
Tracking the unactuated configuration variable
in an underactuated system, in a global sense, has not received much
attention. Here we present a scheme to do so for a  pendubot - a two link robot actuated
only at the first link. We propose a control law that asymptotically tracks any smooth reference trajectory
of the unactuated second joint , from
almost-any initial condition, termed as almost-global asymptotic tracking (AGAT). Further, the actuated joint's
angular velocity remains bounded.
We go on to generalize the proposed scheme  to an n-link system with as many (or more) degrees of actuation
than unactuation, and show that the result holds.
 \end{abstract}

 \section{Introduction}
 Stabilization of underactuated systems about an equilibrium has been exhaustively
 studied (\cite{ortega2002stabilization}, \cite{shiriaev2005constructive} etc.). The underactuated, two-link robot is a frequently encountered example, and perhaps,
one of the first to be studied, in this class of systems. The pendubot (see Figure \ref{fig1}) is a two link robot in which actuation is applied to the first  joint and the second joint is free to rotate. Interchanging the actuated and unactuated joints results in a mechanism termed the 'acrobot'.
A result by Hauser and Murray (\cite{hausermurray}) for the acrobot prompted other investigation into this system.  In \cite{hausermurray}, a {\it local} tracking problem is considered about the inverted equilibrium of the acrobot or the 'swing up' state. The authors use an approximation to the nonlinear model and obtain the control law by linearizing this approximate model around the equilibrium. Later, in \cite{spong1}, the control problem of swing up of the two link robot is considered and a control strategy based on partial feedback linearization is proposed to stabilize the upright equilibrium. Energy based control techniques are applied for the swing up control problem in \cite{spong2}, \cite{astrom}, \cite{arunm}, \cite{arunm2}, \cite{shiriaev2000}. In all these papers, the objective is restricted to stabilization or local tracking about the unstable equilibrium.

A global treatment of the stabilization objective is found as an application in \cite{acosta} and \cite{ortega2002}, where, the problem of almost-global asymptotic stabilization (AGAS) of the
upright equilibrium of the acrobot is addressed. The authors use Interconnection and Damping Assignment Passivity-Based Control (IDA-PBC) to design a state feedback law for AGAS of the equilibrium state. In \cite{olfati}, the underactuated system is converted into a suitable cascade normal form and thereafter existing design methods such as backstepping and forwarding are applied for AGAS of the upright position. However, to the best of our knowledge, the problem of almost-global asymptotic tracking (AGAT) of a smooth trajectory for the
unactuated joint, neither for the
acrobot nor for the pendubot, has been investigated in the literature.

Recent developments in geometric nonlinear control theory have provided us with tools to model and control the dynamical behavior of simple mechanical systems (SMSs). A simple mechanical system (\cite{bulo} and defined in II) comprises of a class of mechanical systems whose dynamics can be completely described by (a) the configuration manifold, (b) the kinetic energy which defines a metric on the manifold, (c)  the set of available control vectors, and (d)
the external forces acting on the system. The underactuated two link robot is a simple mechanical system on $S^1 \times S^1$. AGAS for a fully actuated SMS for which the configuration space is a Lie group has been studied in \cite{kodi}, \cite{dayawansa}, \cite{pidmtp}, \cite{anrnb2}. However, the dynamics of the second link, in isolation, is not an SMS on $S^1$.

The control objective in this article is to asymptotically track a reference trajectory of link 2, which is not actuated. The control torque applied on the link 1 is induced on the link 2 through the coupling mechanism. Interconnected, underactuated mechanical systems have been studied in the context of hoop robots in \cite{madhumtp} and for a rigid body with $3$ internal rotors in \cite{anrnb1}. In these papers the coupled body which is desired to be controlled is made to look like an SMS by employing feedback control. This method is called feedback regularization \cite{madhumtp}.
\\
\textbf{Contribution and organization}\\
In this paper, we express the dynamics of the pendubot as an SMS on $S^1 \times S^1$. Next, we express the error dynamics for the specified tracking variable, and, apply a feedback control law which makes the error dynamics an SMS on $S^1$. The geometric setting employed to describe the dynamics and feedback control is coordinate free, and therefore, global in representation. Finally, we apply the existing AGAT control for an SMS on $S^1$ in \cite{anrnb2} which leads to the tracking objective being achieved. The contribution of this paper is essentially in treating the underactuated two link robot in a purely geometric setting and solve the problem of almost-global tracking of any bounded, smooth reference trajectory on $S^1$.

The paper flows as follows. Section II deals with preliminaries on frequently used notions in theory of Lie groups and the description of an SMS on a Lie group. In Section III, the equations of motion are derived for the pendubot by applying variational principles. The fourth section introduces the tracking problem for the pendubot and a control law is derived for AGAT of a reference trajectory. In Section V we generalize the tracking control for an $n$-link robot with $l$ unactuated joints. The next section is simulation verification of the proposed tracking control for two pendubots.

\section{SMS on a Lie group}
\subsection{Preliminaries on Lie groups}
Let $G$ be a Lie group and let $\mathfrak{g}$ denote its Lie algebra. Let $\phi: G \times G \to G$ be the left group action in the first argument defined as $\phi(g,h) \coloneq  L_{g} (h)= gh $ for all $g$, $h \in G$.  The Lie bracket on $\mathfrak{g}$ is denoted by $[,]$. For matrix Lie groups, the Lie bracket is the commutator operator. $Ad_g : \mathfrak{g} \to \mathfrak{g}$ is defined as $Ad_g(\xi) = T_eL_gR_{g^{-1}}\xi$. It's dual $Ad_g^*: \mathfrak{g}^* \to \mathfrak{g}^*$ is defined as $\langle Ad_g^* \alpha , \eta \rangle \coloneq \langle \alpha, Ad_g \eta \rangle$ for $\alpha \in \mathfrak{g}^*$. The tangent map to $Ad_g$ is called \textit{adjoint map}, and denoted as $ad_\xi : \mathfrak{g} \to \mathfrak{g}$ for $\xi \in \mathfrak{g}$. It is defined as $ad_\xi \eta \coloneq [\xi, \eta]=\frac{\mathrm{d}}{\mathrm{d}t}|_{t=0} Ad_{exp(t\xi)} \eta$ for $\eta \in \mathfrak{g}$. We define the dual  $ad^*_\xi : \mathfrak{g}^* \to \mathfrak{g}^*$ to the adjoint map as $\langle ad^*_\xi \alpha, \eta \rangle = \langle \alpha, ad_{\xi} \eta \rangle $.
Let $\mathbb{I} :\mathfrak{g} \to \mathfrak{g}^*$ be an isomorphism from the Lie algebra to its dual. The inverse is denoted by $\mathbb{I}^\sharp: \mathfrak{g}^* \to \mathfrak{g}$. $\mathbb{I}$ induces a left invariant metric on $G$ (see Section 5.3 in \cite{bulo}), which we denote by $\mathbb{G}_{\mathbb{I}}$ and define by the following $\mathbb{G}_{\mathbb{I}}(g).(X_g,Y_g) = \langle \mathbb{I}(T_gL_{g^{-1}} (X_g)),T_gL_{g^{-1}} (Y_g)\rangle$ for all $g \in G$ and $X_g$, $Y_g \in T_g G$. $\stackrel{\mathfrak{g}}{\nabla}$ is the bilinear map defined as
\begin{equation}\label{eq:60}
 \stackrel{\mathfrak{g}}{\nabla}_\xi \nu = \frac{1}{2} [\xi, \nu] - \frac{1}{2} \mathbb{I}^ \sharp ( ad^*_{\xi} \mathbb{I}\nu +  ad^*_{\nu} \mathbb{I}\xi  )
\end{equation}
for $\nu$, $\xi \in \mathfrak{g}$.

\subsection{SMS on a Lie group}
 A simple mechanical system (or SMS) on a Lie group $G$ with a metric $\mathbb{G}_{\mathbb{I}}$ is denoted by the 7-tuple $(G,\mathbb{G}_{\mathbb{I}},V, F, \mathcal{F})$, where $V: G \to \mathbb{R}$ is a potential function on $G$, $F \in \mathfrak{g}^*$ is an external uncontrolled force, $\mathcal{F}= \{f^1 ,\dotsc, f^m\}$ is a collection of covectors in $\mathfrak{g}^*$. The control forces are covector fields given by $F^i(g) = T^*_g L_{g^{-1}} f^i$, $i = 1\dotsc m$. The SMS is fully actuated if $T_g^*G = span\{F^1(g), \dotsc , F^m(g)\}$, $\forall g \in G$.  The equations of motion for the SMS $(G, {\mathbb{I}},F, \mathcal{F})$ are given by
\begin{align}\label{dynliegrp}
\xi &= T_g L_{g^{-1}} \dot{g},\\ \nonumber
\dot{\xi} - \mathbb{I}^\sharp ad^*_\xi \mathbb{I} \xi &= \mathbb{I}^{\sharp} F
\end{align}
where $g(t)$ describes the system trajectory.

\section{Equations of motion}
\begin{figure}[h!]
  \centering
  \includegraphics[scale=0.8]{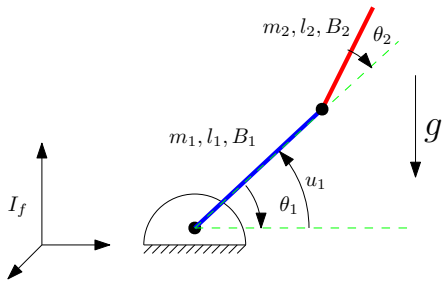}
  \caption{Schematic model of pendubot}\label{fig1}
\end{figure}

Since much of the theory that is employed to synthesize the tracking control law rests on geometric ideas,
our approach to the pendubot will proceed on such geometric lines, and exploit its Lie group structure. At first sight,
this might seem an overkill for the two-link problem, but the logic of the control synthesis is more evident in the
geometric setting. Further, the generalization to the $n$-link manipulator is more easily done employing the
geometric framework.

The pendubot (in Figure \ref{fig1}) is a simple mechanical system which evolves on the manifold $SO(2) \times SO(2)$. To enable a matrix Lie group rooted approach to the problem, all the configuration variables in this article are expressed in $SO(2)$, since $S^1$ is diffeomorphic to $SO(2)$ by the map $f$ defined as:
\[
f: S^1 \to SO(2) \;\;\;\;
f(e^{i \theta}) \coloneq \begin{pmatrix}
\sin \theta & - \cos \theta \\ \cos \theta  & \sin \theta
\end{pmatrix}
\]
where, $e^{i \theta} = \cos \theta + i \sin \theta$ and $e^{i \theta} \in S^1$,
\\
The kinetic energy of the two link robot is considered as the Riemannian metric. The moment of inertia about respective hinge points for link 1 (in blue) and link 2 (in red) are $\mathbb{I}_1$ and $\mathbb{I}_2$ respectively. $B_1$, $B_2$ are body fixed frames on link 1 and link 2 respectively and $I_f$ is the inertial frame. $R_1, R_2 \in SO(2)$ are rotation matrices from $B_1$ to $I_f$ and from $B_2$ to $B_1$ respectively. In Figure \ref{fig1}, $R_i = e^{i \theta_i}$. The center of mass of each rod is assumed to be at the midpoint. $\omega_1$ and $\omega_2$ are body velocities defined as $\hat{\omega}_i= R_i^T \dot{R}_i$ for $i=1,2$. Kinetic energy of the two link robot is given by
\begin{equation}\label{ke}
  K(R_1,R_2,\omega_1,\omega_2) = \int_{B_1} ||\dot{x}_{B_1}||^2 \rho_1 \mathrm{d}V_1+\int_{B_1} ||\dot{x}_{B_2}||^2 \rho_2 \mathrm{d}V_2
\end{equation}
where , $x_{B_1} = R_1 X_{B_1}$, $x_{B_2} = R_1 R_2 X_{B_2} + R_1L_1$, $X_{B_i}$ is the position of a unit mass in the $i$th link in $I_f$ frame and, $L_1= \begin{pmatrix}
  0 & l_1
\end{pmatrix}^T$, $\rho_i$ is the density of $i$th link. Expanding (see Appendix \ref{app1} for details) yields
\begin{align} \label{ke1}
K(\omega_1, \omega_2, R_2) = &\omega_1^2(\mathbb{I}_1 + \mathbb{I}_2 + m_2 l_1^2+m_2L_1^T R_2 L_2 ) \\ \nonumber
&+ \omega_1 \omega_2 ( 2 \mathbb{I}_2+ m_2L_1^T R_2 L_2 )+ \mathbb{I}_2 \omega_2^2
\end{align}
The potential energy is given by
\begin{align}\label{pe}
V(R_1, R_2)= (\frac{m_1}{2}+m_2)g l_1 e_1^T R_1 e_1 + \frac{m_2}{2}g l_2 e_1^T R_2 R_1 e_1
\end{align}
where $e_1 = \begin{pmatrix}
1 & 0
\end{pmatrix}^T$. The equations of motion are (see Appendix \ref{app2})\\
\begin{subequations}
\fbox{
 \addtolength{\linewidth}{-2\fboxsep}%
 \addtolength{\linewidth}{-2\fboxrule}%
 \begin{minipage}{\linewidth}
\begin{equation}\label{dyneq13}
\dot{\omega}_1 = \frac{1}{2K_1}(u_1 - \Gamma_1 - K_2 \dot{\omega}_2 - \alpha(2\omega_1 + \omega_2))
\end{equation}
\begin{align}\label{dyneq23}
(2K_3 - \frac{K_2^2}{2K_1})\dot{\omega}_2 = &-\frac{K_2}{2K_1}(u_1 -\Gamma_1- \alpha(2\omega_1 + \omega_2)) \\ \nonumber
&- \alpha \omega_1 + \beta(\omega_1^2 + \omega_1 \omega_2) - \Gamma_2.
\end{align}
\end{minipage}
}
\end{subequations}
where,
\begin{align*}
&K_1 = \mathbb{I}_1 + \mathbb{I}_2 + m_2 l_1^2+m_2L_1^T R_2 L_2 ,\\
&K_2 =2 \mathbb{I}_2+ m_2L_1^T R_2 L_2 ,\quad  K_3 = \mathbb{I}_2,\\
&\alpha=m_2 \langle L_1 L_2^T, R_2 \hat{\omega}_2 \rangle, \quad \beta = m_2{\{skew(R_2 L_1 L_2^T)\}}\breve{},\\
&\Gamma_1 \coloneq (0.5{m_1}+m_2)g l_1 {\{skew(R_1 e_1 e_1^T)\}}\breve{}\\
&+ 0.5{m_2} g l_2 {\{skew(R_1 R_2^T e_1 e_1^T)\}}\breve{} ,\\
&\Gamma_2 \coloneq 0.5{m_2}g l_2 {\{skew(R_2 e_1 e_1^T R_1^T)\}}\breve{}.
\end{align*}

\section{Tracking control}
In this section, we follow the procedure in Section V of \cite{anrnb2} to define a tracking error and error dynamics for the unactuated joint. Let ${R_{2_{ref}}} :\mathbb{R}^+ \to SO(2)$ be a smooth reference trajectory for joint 2 which has bounded velocity.
The objective is to choose $u_1(t)$ in \eqref{dyneq23} so that $R_2(t)$ tracks ${R_{2_{ref}}}(t)$
from almost all initial conditions with asymptotic convergence (this is almost-global tracking which we define later, in Definition \ref{AGATdef}). The almost-global tracking of $R_{2_{ref}}$ is achieved in two steps.
\begin{itemize}
\item In the first step, we define the error dynamics for the unactuated link which describes the evolution of the error trajectory. We ensure that the control input appears in this
equation by substituting for the actuated variable in terms of the control.
\item  In the second step, in order to bring in an SMS structure to the unactuated dynamics,
we define a new control which incorporates the additional terms resulting due to step 1.
Once the system has the SMS structure, existing tracking laws are easily implemented. We then
proceed to synthesize AGAT control.
\end{itemize}
We shall henceforth refer to this two-step procedure as the separation
principle.

\subsection{AGAS of error dynamics}
Let us denote the angular velocity of the reference trajectory as ${\hat{\omega}}_r \coloneq {R_{2_{ref}}}^T \dot{R}_{2_{ref}}$. We then define a configuration error trajectory$E(t) \coloneq {R_{2_{ref}}}(t) R_2^T(t)$ on $SO(2)$. The velocity of this error trajectory is given as
\begin{equation}\label{edot}
\dot{E}= {R_{2_{ref}}}({\hat{\omega}}_r - {\hat{\omega}}_2) R_2^T
\end{equation}
Next, we define a closed loop energy like function $E_{cl}: TSO(2) \to \mathbb{R}$ as follows
\begin{equation}\label{Ecl}
E_{cl}(E, \dot{E})= K_p \psi(E) + \frac{1}{2} {||\dot{E} ||_{\mathbb{G}_{\mathbb{I}}}}^2
\end{equation}
where $\mathbb{G}_{\mathbb{I}}$ is the metric induced on $SO(2)$ by the $\mathbb{I}$ (see section 5.3.1 in \cite{bulo} for more details), $\mathbb{I} \coloneq 2K_3 - \frac{K_2^2}{2K_1}$, $\psi(E) \coloneq tr\{P(id- E)\}$, $id$ denotes the $2 \times 2$ identity matrix, $P\coloneq diag(c_1,c_2)$, and $ c_1 + c_2 \neq 0$.

\begin{remark} At this stage the problem is entirely focused on $SO(2)$.
\end{remark}
\begin{definition}\label{AGATdef}
The reference trajectory $R_{2_{ref}}$ is \textbf{almost-globally stable with respect to} the closed loop energy like function $E_{cl}$ (defined in \eqref{Ecl}) if, for almost all initial conditions $(R_2(0), \omega_2(0)) \in SO(2) \times \sotwo $, the function $t \to E_{cl}(t)$ is non-increasing.
\end{definition}

\begin{definition}\label{hessdef}
The Hessian of $\psi$ is the symmetric $(0,2)$ tensor field denoted by $Hess \psi$ and defined as $ Hess \psi(q)(v_q, w_q) =  \mathbb{G}_{\mathbb{I}}( v_q, \stackrel{\mathbb{G}_{\mathbb{I}}}{\nabla}_{w_q} grad \psi),$ where $v_q$, $w_q \in T_q SO(2)$ and $grad$ denotes the gradient vector field.
\end{definition}
If $x_0$ is a critical point of $\psi$ and $\theta$ is the local coordinate at $x_0$, then, $Hess(\psi)(x_0)$ in coordinates is $Hess \psi(x_0) = \frac{\partial ^2 \psi}{\partial \theta^2}(x_0)$ (see chapter 13 in \cite{milnor} for details)\\
\begin{definition}\label{navfn}(\cite{kodi})
A function $\psi: SO(2) \to \mathbb{R}$ on $(SO(2), \mathbb{G}_{\mathbb{I}})$ is a navigation function if
\begin{enumerate}
  \item $\psi$ attains a unique minimum.
  \item $Det(Hess \psi(q^*)) \neq 0$ whenever $\mathrm{d}\psi(q^*) = 0$ for some $q^* \in SO(2)$.
\end{enumerate}
\end{definition}
\begin{lemma}\label{lem1}
$\psi$ is a navigation function on $SO(2)$.
\end{lemma}
\begin{proof} We proceed to determine the critical points of $\psi: SO(2) \to \mathbb{R}$.
\begin{align*}
\frac{\mathrm{d}}{\mathrm{d}t} \psi(E) &= - tr(P \dot{E}) = - tr(P E E^T \dot{E})\\
&=- tr\big( (skew(PE)+ sym(PE))(E^T \dot{E})\big)\\
&=  2{\{skew(PE)\}}\breve{} \quad \{E^T \dot{E} \}\breve{}
\end{align*}
The third step follows from the equality $tr(xy) =- \hat{x} \hat{y} $ and the fact that $E^T \dot{E} \in \sotwo$ is skew symmetric. Therefore, the critical points of $\psi$ are the solution to the equation $skew(PE)=0_{2 \times 2}$ or, $PE = E^T P$. Let $E = \begin{pmatrix}
x &y \\ -y &x
\end{pmatrix}$ where $x^2+y^2=1$, therefore, $skew(PE) = \begin{pmatrix}
0 & y(c_1+c_2)\\ -y(c_1+c_2) & 0
\end{pmatrix}$. So, the two critical configurations are given by $\begin{pmatrix}
 1 & 0 \\ 0 & 1
\end{pmatrix}$ and $\begin{pmatrix}
 -1 & 0 \\ 0 & -1
\end{pmatrix}$. Observe that $\psi(E)= tr(P(id-E)) = (1-x)(c_1+c_2)$ and hence, that $id$ is the unique minimum and $-id$ is the unique maximum of $\psi$. It is verified that the Hessian is positive definite at both critical points along the lines of Proposition 11.31 in \cite{bulo}.
\end{proof}
We choose the error dynamics for AGAT of $R_{2_{ref}}$ as the SMS $(SO(2), \mathbb{G}_{\mathbb{I}}, -K_p \mathrm{d}\psi(E) + F_d \dot{E})$, where, $F_d$ is a dissipative force, and, $K_p$ and $\psi$ are defined in \eqref{Ecl}. Then the error dynamics are given by the following equations
\begin{equation} \label{errdyn}
\stackrel{\mathbb{G}_{\mathbb{I}}}{\nabla}_ {\dot{E}} \dot{E}= \mathbb{G}^\sharp _{\mathbb{I}}(-K_p \mathrm{d}\psi(E) + F_d \dot{E})
\end{equation}

\begin{lemma}\label{lem2}
The error dynamics in \eqref{errdyn} is almost-globally asymptotically stable about $(E^*, \dot{E}^*)=(id,0)$
\end{lemma}
\begin{proof}

Observe that  $E_{cl}(id,0)=0$ and $E_{cl}(q,0)>0$ for all $(q,0) \in T SO(2)$ in a neighborhood of $(id,0)$. Also,
\begin{align*}
\frac{\mathrm{d}}{\mathrm{d}t} E_{cl}(E, \dot{E}) &= \langle K_p \mathrm{d} \psi(E),\dot{E}  \rangle + \ll \dot{E}  , \stackrel{\mathbb{G}_{\mathbb{I}}}{\nabla}_{\dot{E}} \dot{E}   \gg \\
&=\langle K_p \mathrm{d} \psi(E),\dot{E} \rangle \\
&+ \mathbb{G}_{\mathbb{I}}(\dot{E},  -\mathbb{G}_{\mathbb{I}}^{\sharp} (K_p \mathrm{d} \psi(E) - F_{d}(\dot{E}  ) )\\
&= K_p \langle \mathrm{d} \psi, \dot{E}   \rangle- K_p \langle \mathrm{d}\psi, \dot{E}   \rangle + \langle F_{d}(\dot{E}  ), \dot{E}   \rangle \leq 0
\end{align*}
as $F_{d}$ is dissipative. Therefore $E_{cl}$ is a Lyapunov function and the error dynamics in \eqref{errdyn} is locally stable around $(id,0)$. The almost-global stability result follows from Lemma 1 in \cite{anrnb2} as $\psi$ is a navigation function, $E(g,g)=id$ for all $g \in SO(2)$ and $(\psi, E)$ is a compatible pair.
\end{proof}
\begin{remark}
The 'zero error' equilibrium state mentioned in the beginning of this section is $(id,0)$. The configuration at which the navigation function achieves its minimum is the 'zero error' configuration and the RHS of \eqref{errdyn} is the control vector field which drives the error trajectory to $(id,0)$.
\end{remark}

\subsection{Separation principle and AGAT of $R_{2_{ref}}(t)$}
The separation principle has two important steps. For the first step, the LHS of \eqref{errdyn} is expressed in terms of the trajectories $R_2(t)$, $R_{2_{ref}}(t)$ and their velocities. For the second step the feedback terms to be introduced through $u_1$ are identified. The following theorem states the main result of this paper.
\begin{theorem}\label{thm}
Consider the pendubot described by equations \eqref{dyneq13}-\eqref{dyneq23}, the smooth bounded reference trajectory ${R_{2_{ref}}}:\mathbb{R}^+ \to SO(2)$ is almost-globally asymptotically tracked by the control law
\begin{align} \label{thmeqn}
u_1 = &\frac{2K_1}{K_2} \bigg\{ \underbrace{ \{Ad_{R_2}^*u\}\breve{}}_{\text{PD term}}\\ \nonumber
&+ \underbrace{\frac{K_2 }{2 K_1}(\alpha(2\omega_1 + \omega_2) +\Gamma_1) - \alpha \omega_1 + \beta(\omega_1^2 + \omega_1 \omega_2) -\Gamma_2}_{\text{cancelling the quadratic terms}}\\ \nonumber
&- \underbrace{ \{ \mathbb{I} (\dot{{\hat{\omega}}}_r +[{\hat{\omega}}_2,{\hat{\omega}}_r])\} \breve{} \quad }_
{\text{feedforward terms}}\bigg\}
\end{align}
where $R_2(t)$ is the controlled trajectory for link 2, and $u \in \mathfrak{g}^*$ is defined as $u= E^T(-2 K_p skew(PE) + F_d \dot{E})$.
\end{theorem}
\begin{proof}
Let the body velocity of the error trajectory be defined as $\eta \coloneq T_E L_{E^{-1}} \dot{E}$. The LHS of \eqref{errdyn} can be simplified as follows
\begin{align}\label{errdyn1}
\stackrel{\mathbb{G}_{\mathbb{I}}}{\nabla}_ {\dot{E}} \dot{E}
&=T_{id} L_E \{ \frac{\mathrm{d}}{\mathrm{d}t} (E^T \dot{E})+ \stackrel{\sotwo}{\nabla}_\eta \eta \}\\ \nonumber
&= T_{id} L_E \{ - R_2 \dot{\hat{\omega}}_2 R_2^T +R_2 ([\hat{\omega}_2, \hat{\omega}_r] -\dot{\hat{\omega}}_2 ) R_2^T + \stackrel{\sotwo}{\nabla}_\eta \eta \} \\ \nonumber
&=  T_{id}  L_E \{- R_2 \dot{\omega}_2 R_2^T +  R_2 ([\hat{\omega}_2, \hat{\omega}_r] +\dot{\hat{\omega}}_r ) R_2^T \}
\end{align}
The first equality follows from Lemma 3 in \cite{anrnb1}. The second equality is given by the following simplification
\begin{align*}
\frac{\mathrm{d}}{\mathrm{d}t} (E^T \dot{E})&= \frac{\mathrm{d}}{\mathrm{d}t}\{ E^T R_{2_{ref}} ({\hat{\omega}}_r - {\hat{\omega}}_2  )R_2^T\}\\
&= \frac{\mathrm{d}}{\mathrm{d}t} \{ R_2({\hat{\omega}}_r - {\hat{\omega}}_2  )R_2^T \}\\
&= \dot{R}_2 \hat{\omega}_r R_2^T + R_2 \dot{\hat{\omega}}_r R_2^T - R_2 \hat{\omega}_r R_2^T \dot{R}_2 R_2^T\\
&-\dot{R}_2 \hat{\omega}_2 R_2^T - R_2 \dot{\hat{\omega}}_2 R_2^T + R_2 \hat{\omega}_2 R_2^T \dot{R}_2 R_2^T\\
&= R_2 (\dot{\hat{\omega}}_r +[\hat{\omega}_2, \hat{\omega}_r] -\dot{\hat{\omega}}_2 ) R_2^T
\end{align*}
The third equality follows from the fact that $SO(2)$ is Abelian, and therefore, $\stackrel{\mathfrak{g}}{\nabla}_\eta \eta=0$.
\newline
Observe (from \eqref{errdyn}) that $\stackrel{\mathbb{G}_{\mathbb{I}}}{\nabla}_ {\dot{E}} \dot{E} = T_{id} L_E \mathbb{I}^\sharp (u)$, where $u = E^T(-2 K_p skew(PE) + F_d \dot{E})$ is the stabilizing control. Now from \eqref{errdyn1},
\[\mathbb{I}^\sharp (u) = - R_2 \dot{\omega}_2 R_2^T +   R_2 ([\hat{\omega}_2, \hat{\omega}_r] +\dot{\hat{\omega}}_r ) R_2^T
\]
Rearranging the above equation leads to
\begin{align}\label{iomegadot}
\mathbb{I} \dot{\omega}_2 &= -Ad_{R_2}^*(u - \mathbb{I} Ad_{R_2}([\hat{\omega}_2, \hat{\omega}_r] -\dot{\hat{\omega}}_2 ) )\\ \nonumber
&=  -Ad_{R_2}^*u + \mathbb{I}(\dot{{\hat{\omega}}}_r +[{\hat{\omega}}_2,{\hat{\omega}}_r])
\end{align}
However, from \eqref{dyneq23} we have,
\begin{align}\label{iomd}
\mathbb{I}\dot{\omega}_2 = &-\frac{K_2}{2K_1}(u_1) +\frac{K_2 }{2K_1} (\alpha(2\omega_1 + \omega_2) + \Gamma_1)- \alpha \omega_1 \\ \nonumber
&+ \beta(\omega_1^2 + \omega_1 \omega_2)- \Gamma_2
\end{align}
Comparing  \eqref{iomegadot} and \eqref{iomd} yields the expression for $u_1$ as in \eqref{thmeqn}. Essentially, we cancel out the last three terms of \eqref{iomd} and introduce the proportional derivative (or PD) control in \eqref{iomegadot} through $-\frac{K_2}{2K_1}(u_1)$. By Lemma 2, the error dynamics for the tracking problem (in \eqref{errdyn}) is almost-globally asymptotically stable about $(id,0)$. Therefore, $u_1$ is the desired control for almost-global tracking of ${R_{2_{ref}}}(t)$.
\end{proof}

% \stackrel{\mathbb{G}_{\mathbb{I}}}{\nabla}_ {\dot{E}} \dot{E}= T_e L_E(\frac{\mathrm{d}}{\mathrm{d}t} (E^T \dot{E}))
 \begin{remark}
In the control law \eqref{thmeqn}, the terms $\{Ad_{R_2}^*u\}\breve{} -  \{ \mathbb{I} (\dot{{\hat{\omega}}}_r +[{\hat{\omega}}_2,{\hat{\omega}}_r])\} \breve{}$ denote the PD plus feed-forward terms for AGAS of the error dynamics. The remaining terms are the quadratic terms that essentially cancel out identical terms in \eqref{iomd} and impart the SMS structure to the error dynamics (in \eqref{iomegadot}).
\end{remark}
\begin{remark}
Note that the cancellation of the quadratic terms is not to be mistaken to be {\it feedback linearization}; it is an
operation solely from a mechanics point of view to impart a mechanical structure to the error dynamics. This feature is
elaborated in  \cite{madhumtp}.
\end{remark}
\section{AGAT of unactuated links of an $n$-link robot}
Consider an $n$-link robot with $m$ actuated (or active) links and $l$ unactuated (or passive) links. Let $q_1 \in \underbrace{SO(2) \times \dotsc \times SO(2)}_\text{m times}$ and $q_2 \in \underbrace{SO(2) \times \dotsc \times SO(2)}_\text{l times}$ denote the generalized variables of active and passive links respectively. The equations of motion can be expressed as
\begin{subequations}
\fbox{
 \addtolength{\linewidth}{-2\fboxsep}%
 \addtolength{\linewidth}{-2\fboxrule}%
 \begin{minipage}{\linewidth}
\begin{equation} \label{nlink1}
M_{11} \ddot{q}_1 + M_{12} \ddot{q}_2 + h_1 + \Phi_1= u_1
\end{equation}
\begin{equation}\label{nlink2}
M_{21} \ddot{q}_1 + M_{22} \ddot{q}_2 + h_2 + \Phi_2= 0
\end{equation}
\end{minipage}
}
\end{subequations}
where $M_{11} \in \mathbb{R}^{m \times m}$, $M_{12} \in \mathbb{R}^{m \times l}$, $M_{21} \in \mathbb{R}^{l \times m}$ and $M_{22} \in \mathbb{R}^{l \times l}$. $h_1$ and $h_2$ contain Coriolis and centrifugal terms (quadratic velocity terms), $\Phi_1$, $\Phi_2$ contain gravitational terms and, $u_1 \in T_{q_1}\underbrace{SO(2) \times \dotsc \times SO(2)}_\text{m times}$ is the control vector field to be chosen such that $q_2(t)$ tracks a given reference trajectory $q_{2_{ref}}(t)$. The equations \eqref{nlink1}-\eqref{nlink2} can be simplified as follows
\begin{subequations}
\begin{equation}\label{nlink3}
\ddot{q}_1= M_{11}^{-1} (u_1 - \Phi_1 -h_1 -M_{12}\ddot{q}_2)
\end{equation}
\begin{align}\label{nlink4}
(M_{22}- M_{21}M_{11}^{-1}M_{12}) \ddot{q}_2 &-M_{21}M_{11}^{-1} (h_1+ \Phi_1) +h_2 + \Phi_2 \\ \nonumber
&=- M_{21}M_{11}^{-1} u_1
\end{align}
\end{subequations}
\begin{remark}
$\mathbb{I} \coloneq M_{22}- M_{21}M_{11}^{-1}M_{12}$ is invertible as it is the Schur complement of the invertible
mass matrix $M \coloneq \begin{pmatrix}
M_{11} & M_{12}\\ M_{21} & M_{22}
\end{pmatrix} $.
\end{remark}

Let $q_{2_{ref}} (t) \in G$ denote the reference trajectory, where $G=\underbrace{SO(2) \times \dotsc \times SO(2)}_\text{l times}$. The group operation is defined component-wise as $L_{(g_{1},\dotsc, g_{l})} (h_{1}, \dotsc, h_{l})= (g_1h_1, \dotsc, g_lh_l)$. Similar to Section IV, we first express the error dynamics for the tracking problem. Let $E= q_{2_{ref}} (t) q_2^{-1}(t)$ be the error trajectory. The error is identity when the $q_2(t)$ coincides with $q_{2_{ref}} (t)$. The error dynamics is given by \eqref{errdyn}. Simplifying the LHS of \eqref{errdyn}, we have,
\begin{align}\label{errdyn2}
\stackrel{\mathbb{G}_{\mathbb{I}}}{\nabla}_ {\dot{E}} \dot{E}
&=T_{id} L_E \{ \frac{\mathrm{d}}{\mathrm{d}t} (E^T \dot{E})+ \stackrel{\mathfrak{g}}{\nabla}_\eta \eta \}\\ \nonumber
&= T_{id} L_E \{ \ddot{q}_2 q_2^{-1} + {\stackrel{\mathfrak{g}}{\nabla}}_\eta \eta + \frac{\mathrm{d}}{\mathrm{d}t}(\dot{q}_{2_{ref}} q_{2_{ref}}^{-1})  \}
\end{align}
where $\eta \coloneq E^{-1} \dot{E}$. In the following theorem we state the necessary condition for AGAT of $q_{2_{ref}} (t)$.
\begin{theorem} \label{thm2}
Consider the $n$-link pendubot described by the equations \eqref{nlink1}-\eqref{nlink2} and the smooth bounded reference trajectory $q_{2_{ref}} : \mathbb{R} \to  G$. If $rank(M_{21}) \geq l$ ($M_{21}$ has full column rank), then the following $u_1$ is the control vector field so that $q_2(t)$ almost-globally and asymptotically tracks $q_{2_{ref}}(t)$.
\begin{align}\label{thm2eqn}
u_1 &= M_{11} M_{21}^{-1}(T_e R_{q_2}(u -  \mathbb{I}\{{\stackrel{\mathfrak{g}}{\nabla}}_\eta \eta + \frac{\mathrm{d}}{\mathrm{d}t}(\dot{q}_{2_{ref}} q_{2_{ref}}^{-1})\})\\ \nonumber
&+ h_2- M_{21}M_{11}^{-1} h_1)
\end{align}
where $u = E^T(-K_p \mathrm{d}\phi(E) + F_d \dot{E})$ and $\phi$ is a navigation function on $G$.
\end{theorem}
\begin{proof}
From \eqref{errdyn}, we have $\stackrel{\mathbb{G}_{\mathbb{I}}}{\nabla}_ {\dot{E}} \dot{E}= T_{id} L_E \mathbb{I}^\sharp u$, where $u = E^T(-K_p \mathrm{d}\psi(E) + F_d \dot{E})$ and $\phi$ is a navigation function on $G$. Substituting $T_{id} L_E \mathbb{I}^\sharp u$ in the LHS of \eqref{errdyn2}, we obtain
\begin{equation}\label{sepstep1}
 u =  \mathbb{I}\ddot{q}_2 q_2^{-1} + \mathbb{I}\{{\stackrel{\mathfrak{g}}{\nabla}}_\eta \eta + \frac{\mathrm{d}}{\mathrm{d}t}(\dot{q}_{2_{ref}} q_{2_{ref}}^{-1})\}
\end{equation}
Substituting for $\mathbb{I} \ddot{q}_2$ from \eqref{nlink4}, we have,
\begin{align}\label{sepstep2}
u&=(M_{21}M_{11}^{-1} h_1 -h_2 + M_{21}M_{11}^{-1} u_1) q_2^{-1}\\ \nonumber
&+ \mathbb{I}\{{\stackrel{\mathfrak{g}}{\nabla}}_\eta \eta + \frac{\mathrm{d}}{\mathrm{d}t}(\dot{q}_{2_{ref}} q_{2_{ref}}^{-1})\}
\end{align}
This implies,
\begin{align}
M_{21}M_{11}^{-1} u_1 &= T_e R_{q_2}(u-  \mathbb{I}\{{\stackrel{\mathfrak{g}}{\nabla}}_\eta \eta + \frac{\mathrm{d}}{\mathrm{d}t}(\dot{q}_{2_{ref}} q_{2_{ref}}^{-1})\})\\ \nonumber
&+ h_2+ \Phi_2 - M_{21}M_{11}^{-1} (h_1 + \Phi_1).
\end{align}
As, $M_{21}$ has full column rank, therefore, $M_{21}^{-1}$ exists, hence,
\begin{align}
u_1 = &M_{11} M_{21}^{-1}(T_e R_{q_2}(u -  \mathbb{I}\{{\stackrel{\mathfrak{g}}{\nabla}}_\eta \eta + \frac{\mathrm{d}}{\mathrm{d}t}(\dot{q}_{2_{ref}} q_{2_{ref}}^{-1})\})\\ \nonumber
&+ h_2 + \Phi_2- M_{21}M_{11}^{-1} (h_1+\Phi_1)).
\end{align}
Therefore, using a separation principle, we have introduced terms through $u_1$, so that the error dynamics in \eqref{errdyn} retains its SMS structure. As $\phi$ is a navigation function, by a proof similar to Lemma \ref{lem2}, the error dynamics is AGAS about $(e,0)$ where $e$ is the identity of $G$. Therefore, $u_1$ is the desired control for AGAT of $q_{2_{ref}}(t)$.
\end{proof}
\begin{remark}
The condition imposed on $M_{21}$ is called 'strong inertial coupling' which means that the number of active degrees of freedom must be atleast as many as the passive degrees in order to solve the tracking problem. This is also observed by the authors in \cite{spong1994pfl}. However, in \cite{spong1994pfl}, the authors use a partial feedback linearization approach to achieve \textit{local} tracking, which means the initial conditions of the system are close to the initial conditions of the reference trajectory. In both Theorems \ref{thm} and \ref{thm2}, the initial conditions of the unactuated joint(s) of the pendubot are allowed to lie in a dense set in $T(SO(2)) $ and $TG$ respectively.
\end{remark}
\begin{remark}
The navigation function $\phi$ in Theorem \ref{thm2} which is necessary for AGAS of the error dynamics can be obtained by composing $\psi$ (in $SO(2)$) $l$ times, as elaborated in Section 3 in \cite{cowan}.
\end{remark}
\section{Simulation results}
For the purpose of simulation, the parameters of the pendubot are: $m_1=0.25 kg$, $m_2=0.2 kg$, $l_1=l_2= 0.5 m$, $g = 9.8 ms^{-2}$ with the following initial conditions:
\begin{equation*}
 R_1(0)= \begin{pmatrix} 1 & 0 \\ 0 & 1 \end{pmatrix}
, \;\;  \omega_1(0) =-1 \quad rad s^{-1},
  \end{equation*}
  \begin{equation*}
   R_2(0)= \begin{pmatrix} 0.4794 & 0.87758 \\ -0.87758 & 0.4794 \end{pmatrix}
 ,\;\; \omega_2(0) = 2 \quad rad s^{-1}
 \end{equation*}
 The reference trajectory is given by
\begin{align*}
	 &R_{2_{ref}}(t) = \begin{pmatrix} \sin(t+\pi/4) & \cos(t+\pi/4) \\ -\cos(t+\pi/4) & \sin(t+\pi/4) \end{pmatrix}, \; \\
     &\omega_{2_{ref}}(t) =1 \quad rad s^{-1}, \; t \geq 0.
\end{align*}
 The AGAT control in \eqref{thmeqn} is implemented by considering $K_p= 2.3$, $F_d=- diag(1.5 \quad 2)$, $P = diag(1 \quad 1.5)$. In this simulation and all others to follow, the differential equations are solved using MATLAB 2014a and the trajectory for the second link (in blue) is compared with the reference trajectory (in red). The first two subfigures plot $R_2(11)= cos(\theta_1)$ with $R_{2_{ref}}(11)$ and $R_2(12)= sin(\theta_1)$ with $R_{2_{ref}}(12)$ respectively. Therefore, they correspond to $x$ and $y$ coordinates in a circle embedded in $\mathbb{R}^2$. The third subfigure shows the first coordinate $E(11)$ of $E(t) \in SO(2)$. As, $E(t) \to id$, it is seen that $E(11)$ converges to $1$. The fourth subfigure shows control effort $|u_1|$. In the fifth subfigure, we plot $\omega_1(t)$ and observe that it remains bounded. The first simulation results are shown in Figure ~\ref{fig2ch3}.\\
\begin{figure}[!h]
\centering
\begin{subfigure}[b]{0.22\textwidth}
  \includegraphics[height=2cm, width=\linewidth]{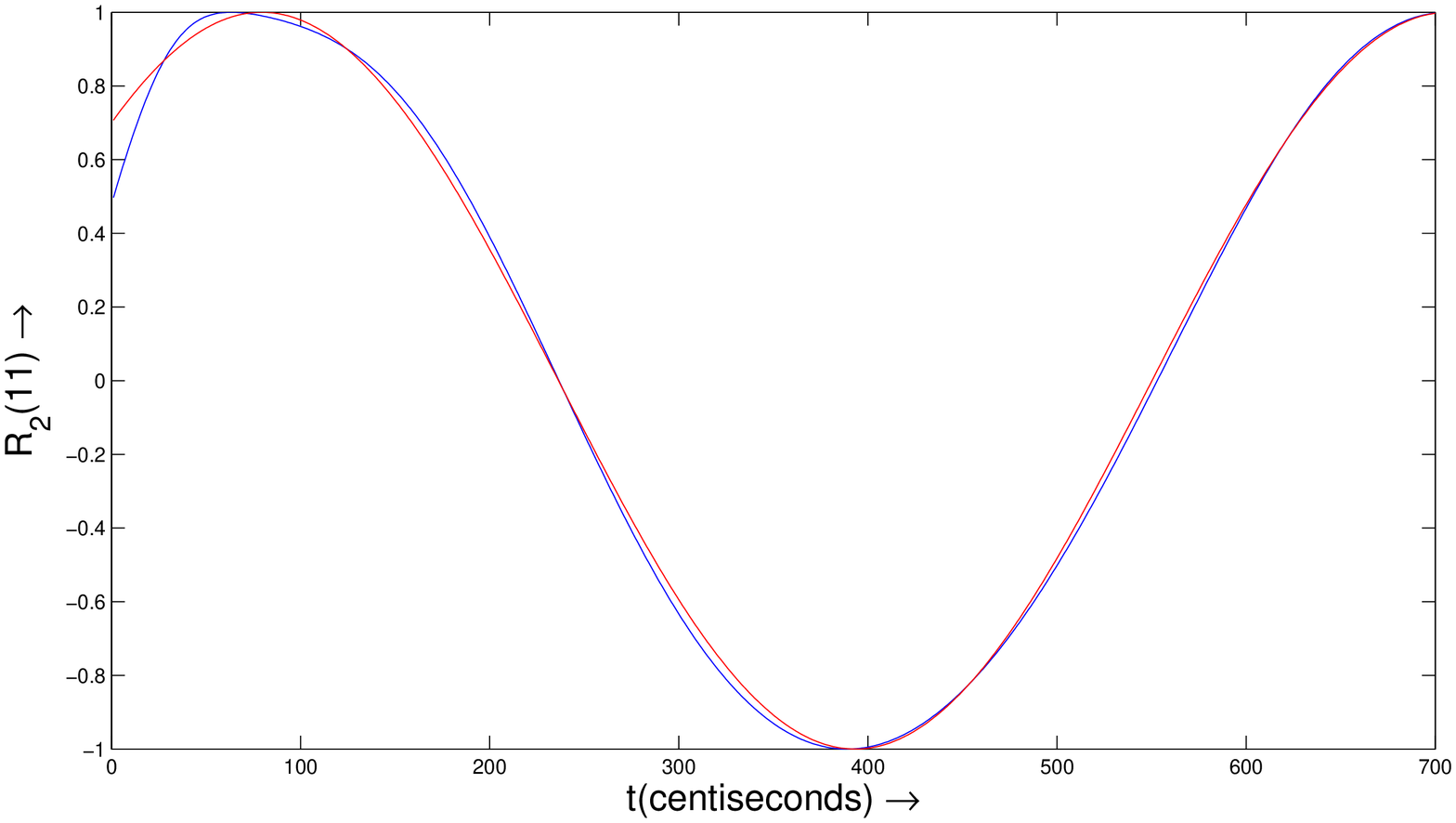}
  \end{subfigure}
\begin{subfigure}[b]{0.22\textwidth}
  \includegraphics[height=2cm, width=\linewidth]{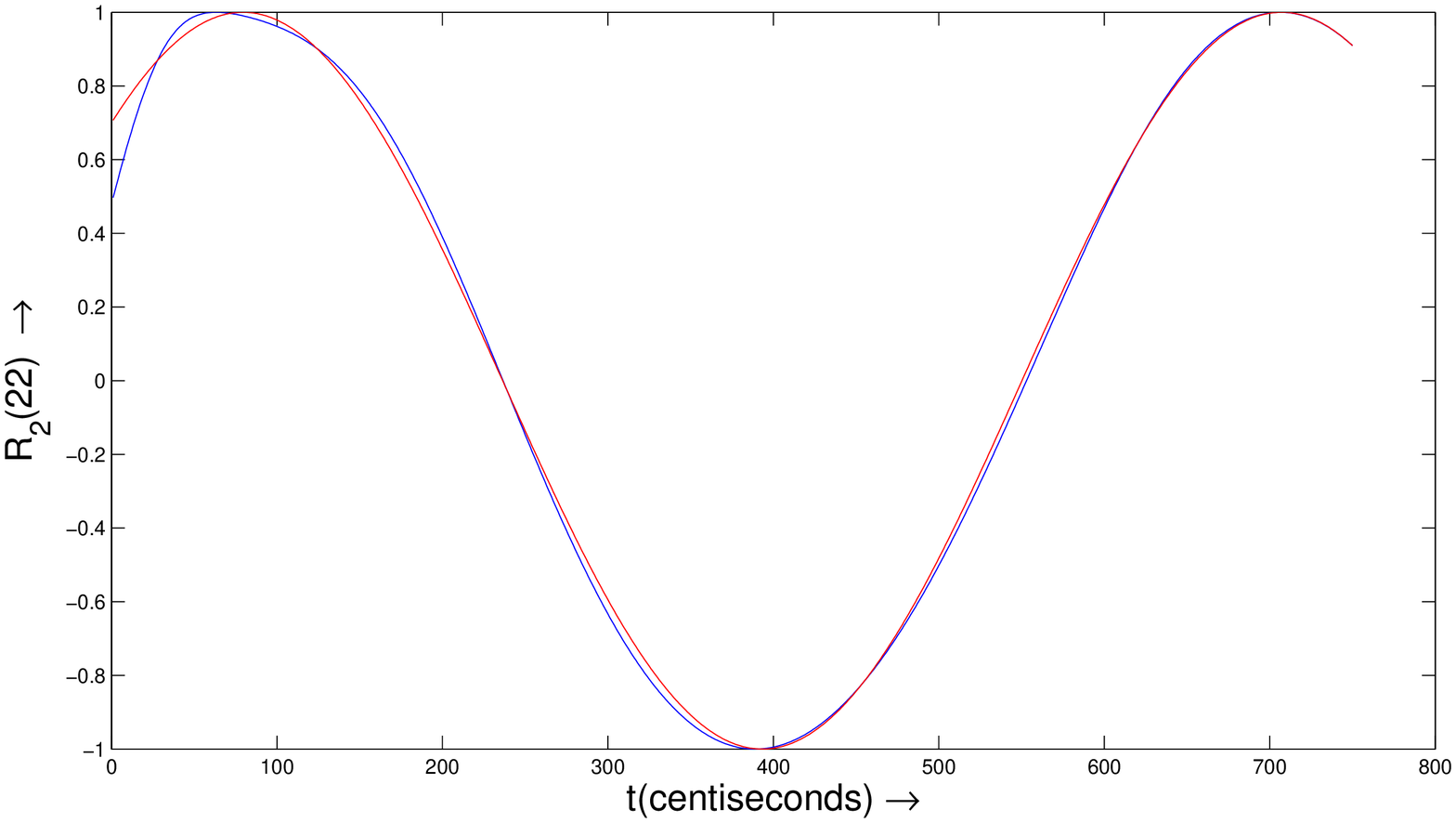}
\end{subfigure}
 \begin{subfigure}[b]{0.22\textwidth}
  \includegraphics[height=2cm, width=\linewidth]{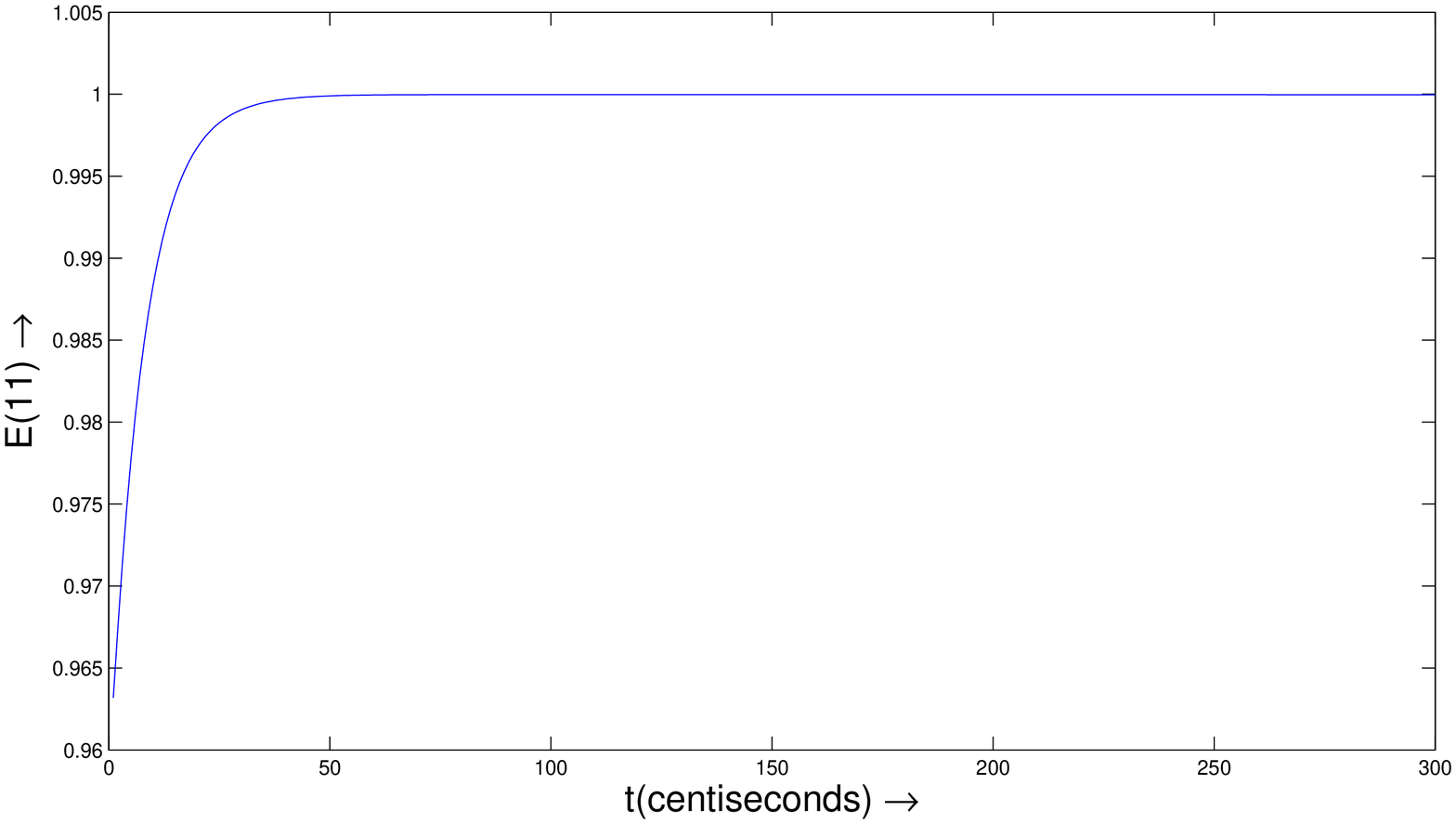}
\end{subfigure}
\begin{subfigure}[b]{0.22\textwidth}
  %\centering
  \includegraphics[height=2cm, width=\linewidth]  {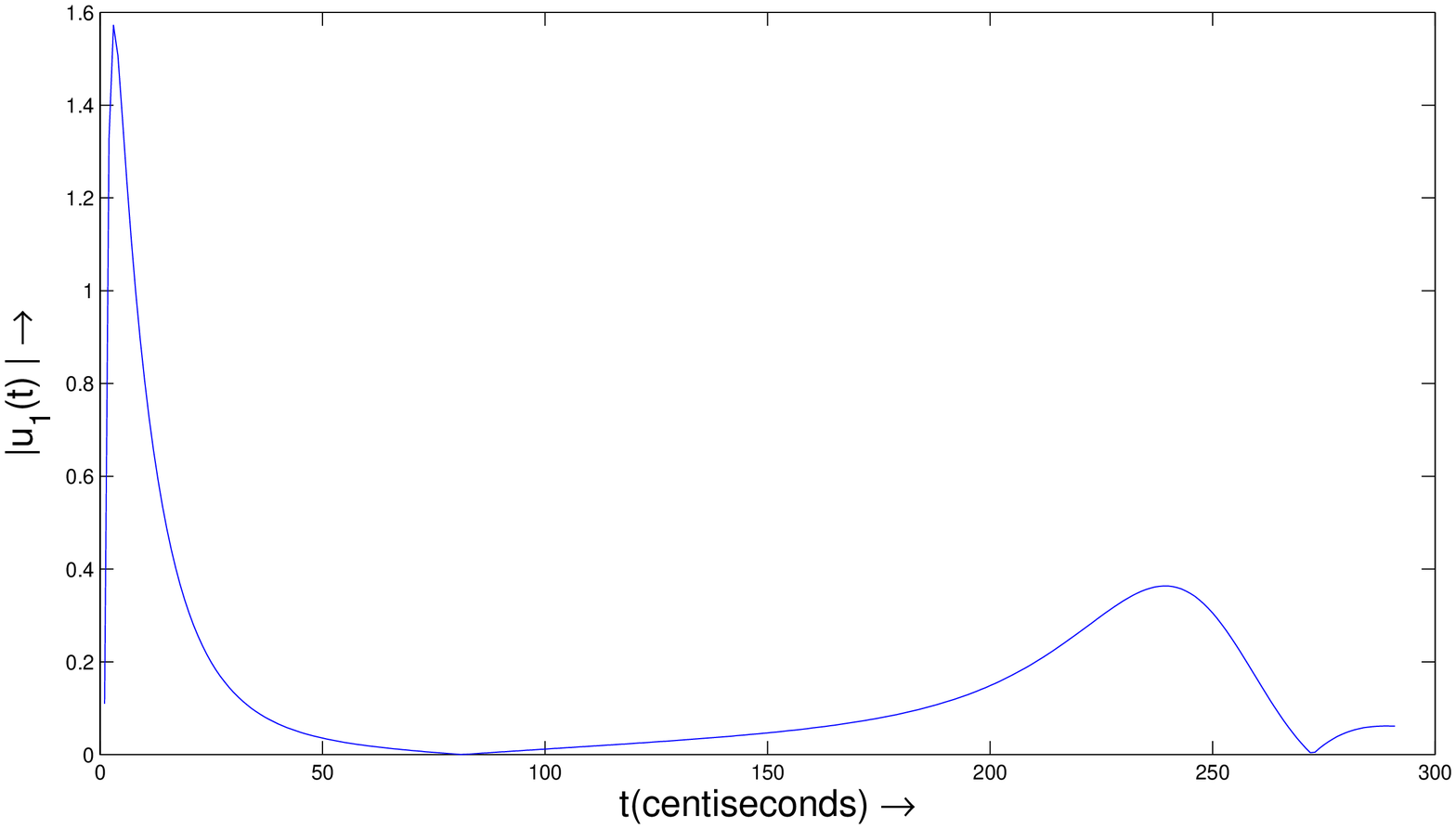}
   \end{subfigure}
   \begin{subfigure}[b]{0.22\textwidth}
  %\centering
  \includegraphics[height=2cm, width=\linewidth]  {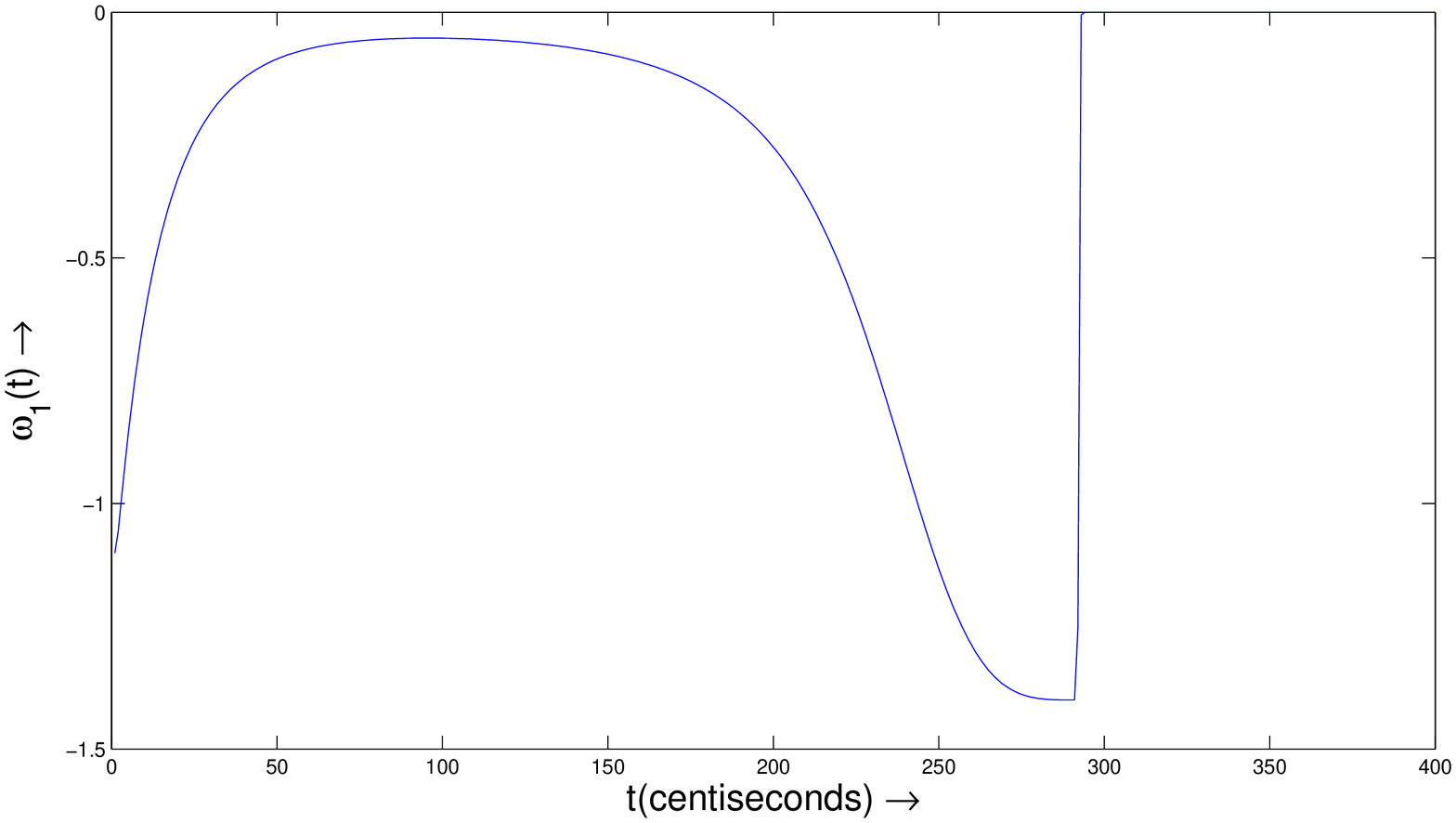}
   \end{subfigure}
 \caption{Tracking results for the first two coordinates for the initial condition $R_2(0)= \protect \begin{pmatrix} 0.4794 & 0.87758 \\ -0.87758 & 0.4794 \protect \end{pmatrix}$ in }
 \label{fig2ch3}
\end{figure}
For the second simulation we choose the following initial conditions:
\begin{equation*}
 R_1(0)= \begin{pmatrix} 1 & 0 \\ 0 & 1 \end{pmatrix}, \;\;\; \omega_1(0)= -1 \quad rad s^{-1}
 \end{equation*}
 \begin{equation*}
 	R_2(0)= \begin{pmatrix} 0 & 1 \\ -1 & 0 \end{pmatrix} , \;\;\; \omega_2(0) = 2 \quad rad s^{-1}
\end{equation*}
The reference trajectory is chosen to be:
\begin{equation*}
R_{2_{ref}}(t) = \begin{pmatrix} \cos(t) & -\sin(t) \\ \sin(t) & \cos(t) \end{pmatrix}, \;\; \omega_{2_{ref}}(t) =1 \quad rad s^{-1}, \;\; t \geq 0.
\end{equation*}
We consider $K_p= 3$, $F_d=- diag(1.5 \quad 2)$, $P = diag(2 \quad 1.5)$. Figure \ref{fig3ch3} shows the results. The last subfigure depicts the reference and the actual trajectory on a cylinder

 Now for the same initial conditions for the pendubot, we consider
\begin{equation*}
R_{2_{ref}}(t) = id, \; \; \omega_{2_{ref}}(0) =0 \quad rad s^{-1}, \; \forall t \geq 0
\end{equation*}
This is therefore, a stabilization problem. The results are shown in Figure \ref{fig32ch3}.
\begin{figure}[!h]
\centering
\begin{subfigure}[b]{0.22\textwidth}
  \includegraphics[height=2cm, width=\linewidth]{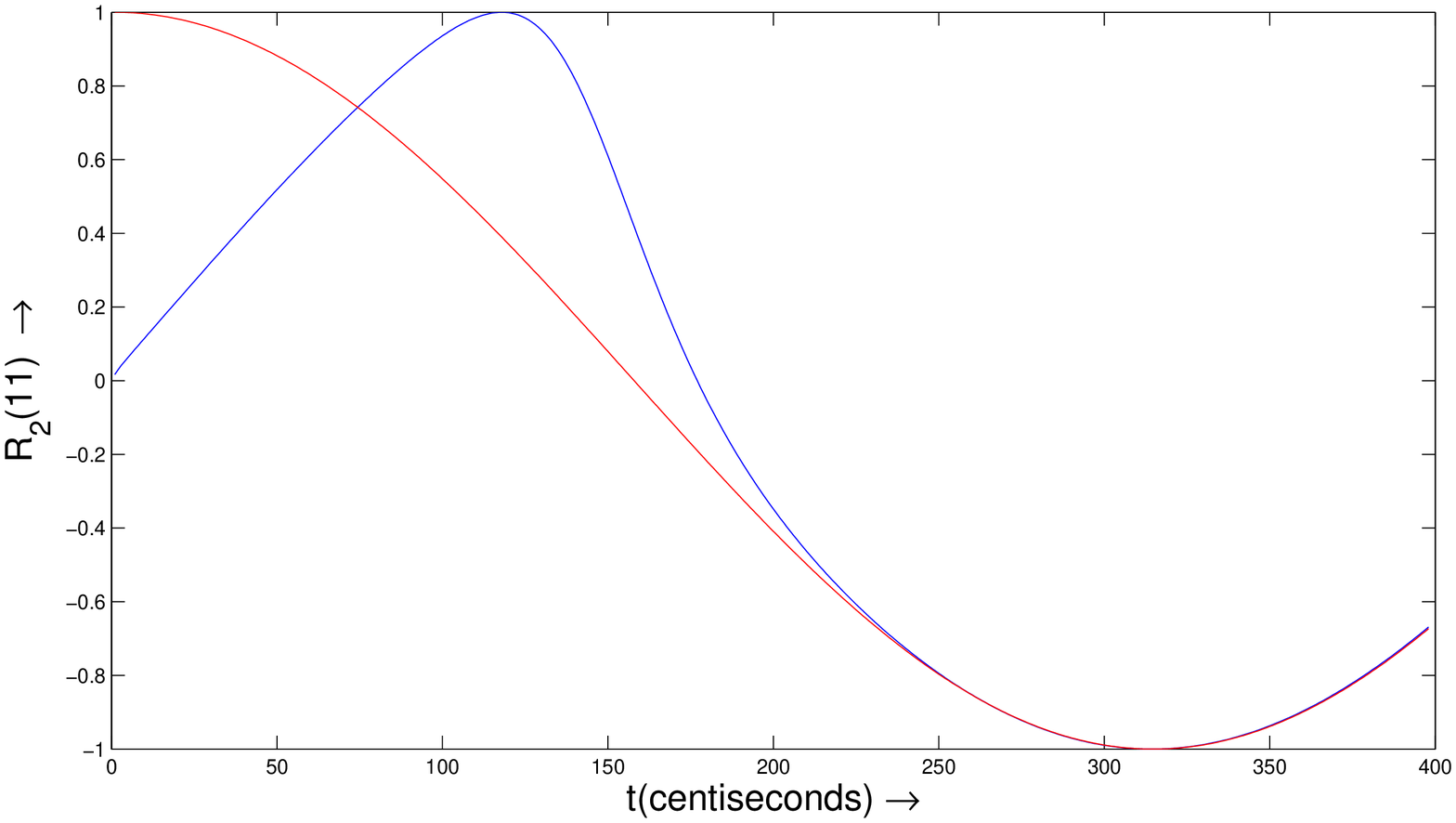}
  \end{subfigure}
\begin{subfigure}[b]{0.22\textwidth}
  \includegraphics[height=2cm, width=\linewidth]{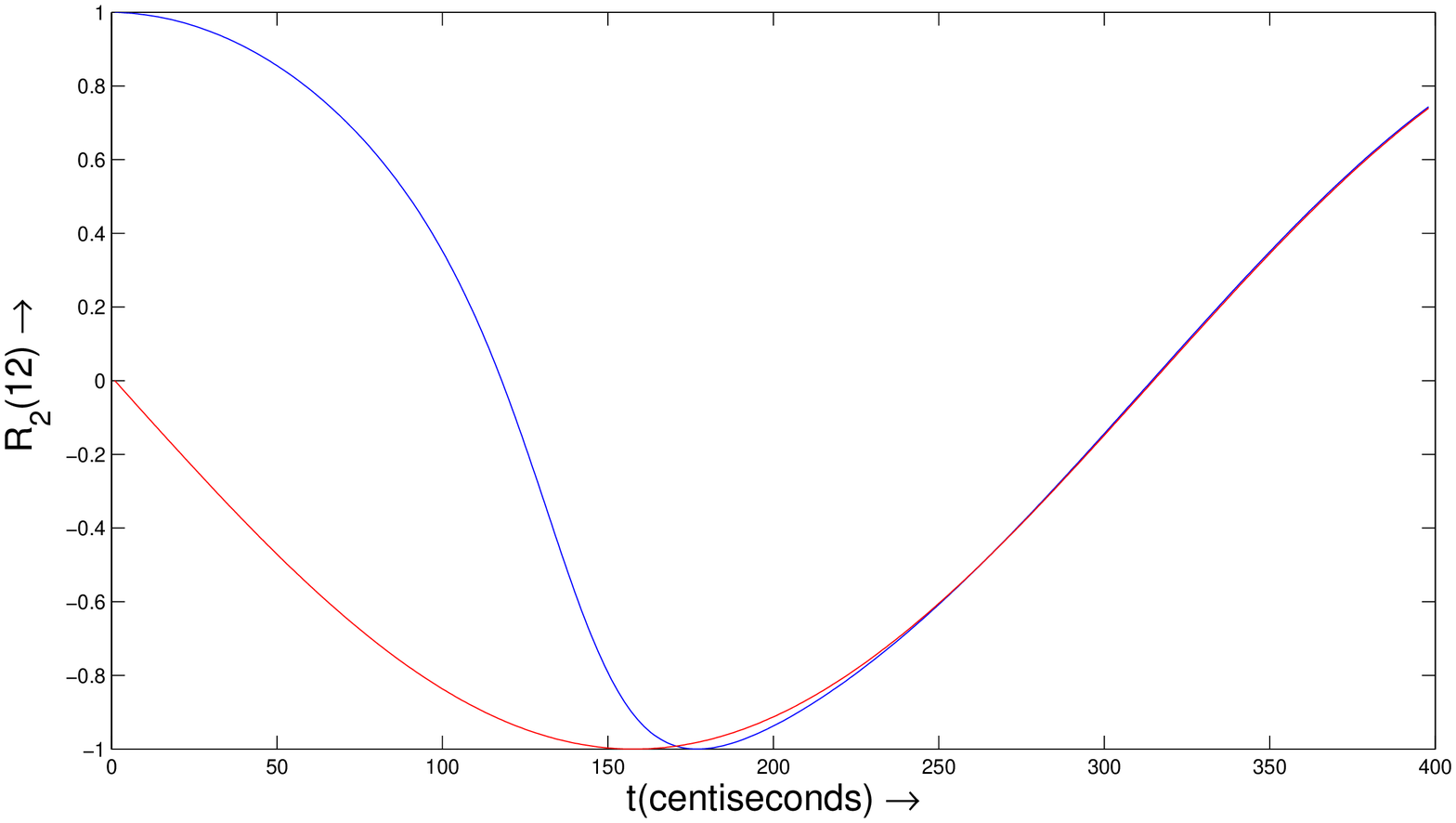}
\end{subfigure}
% \begin{subfigure}[b]{0.22\textwidth}
%   \includegraphics[height=2cm, width=\linewidth]{2link_fig23.eps}
% \end{subfigure}
  \begin{subfigure}[b]{0.22\textwidth}
   %\centering
  \includegraphics[height=2cm, width=\linewidth]  {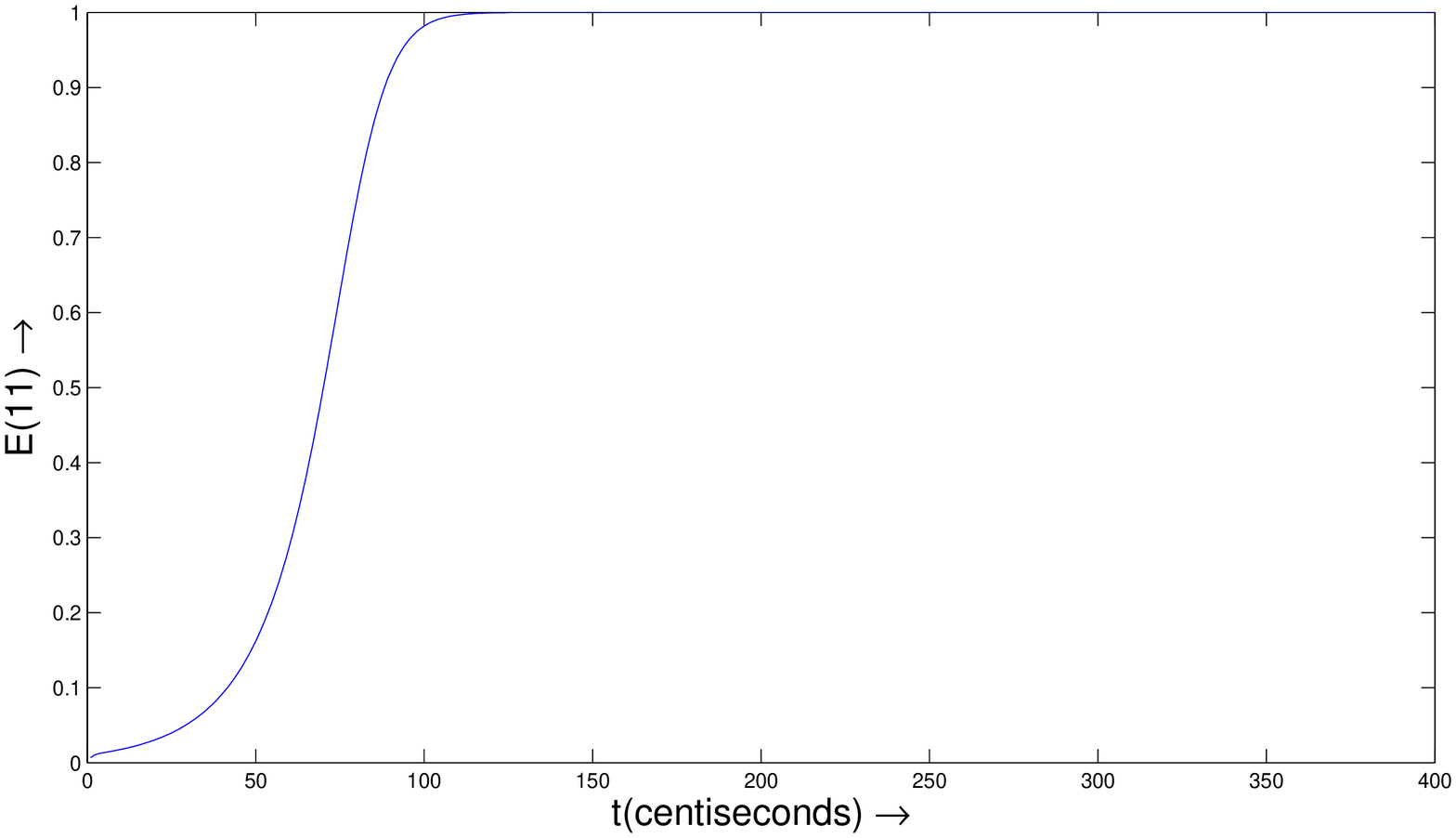}
  \end{subfigure}
  \begin{subfigure}[b]{0.22\textwidth}
   %\centering
  \includegraphics[height=2cm, width=\linewidth]  {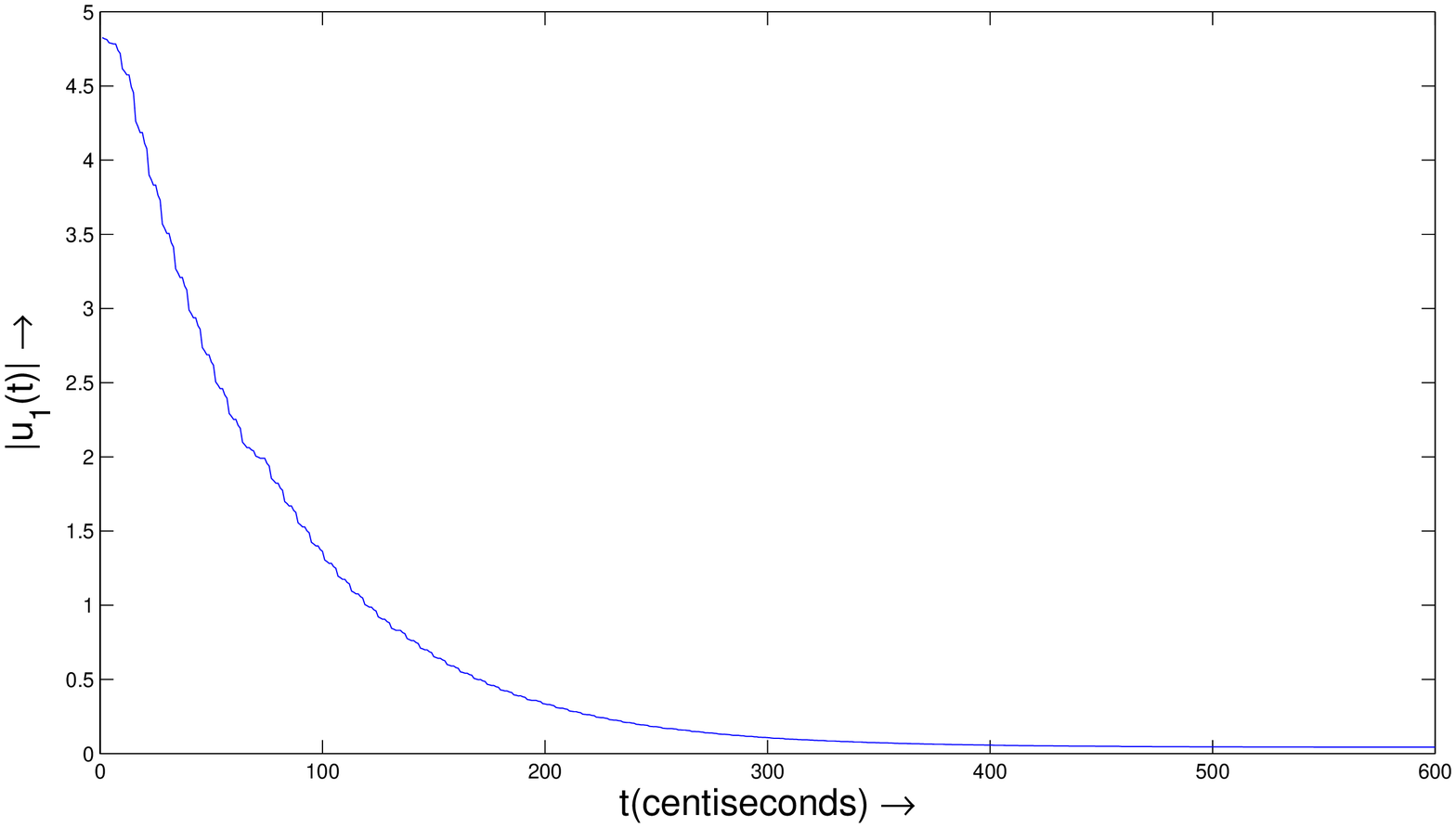}
  \end{subfigure}
  \begin{subfigure}[b]{0.22\textwidth}
   %\centering
  \includegraphics[height=2cm, width=\linewidth]  {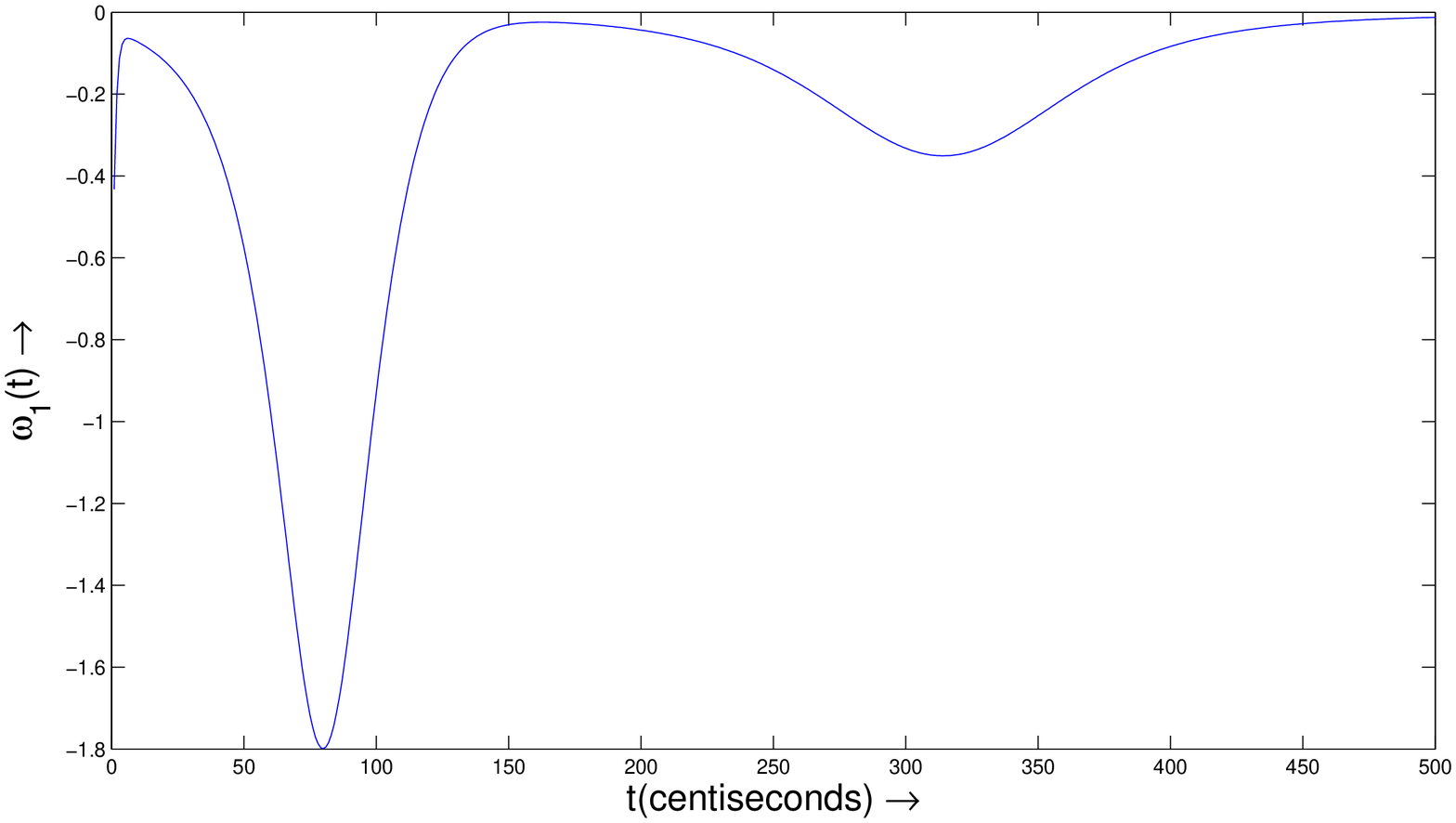}
  \end{subfigure}
  \begin{subfigure}[b]{0.2\textwidth}
   %\centering
  \includegraphics[height=2.7cm, width=1\linewidth]  {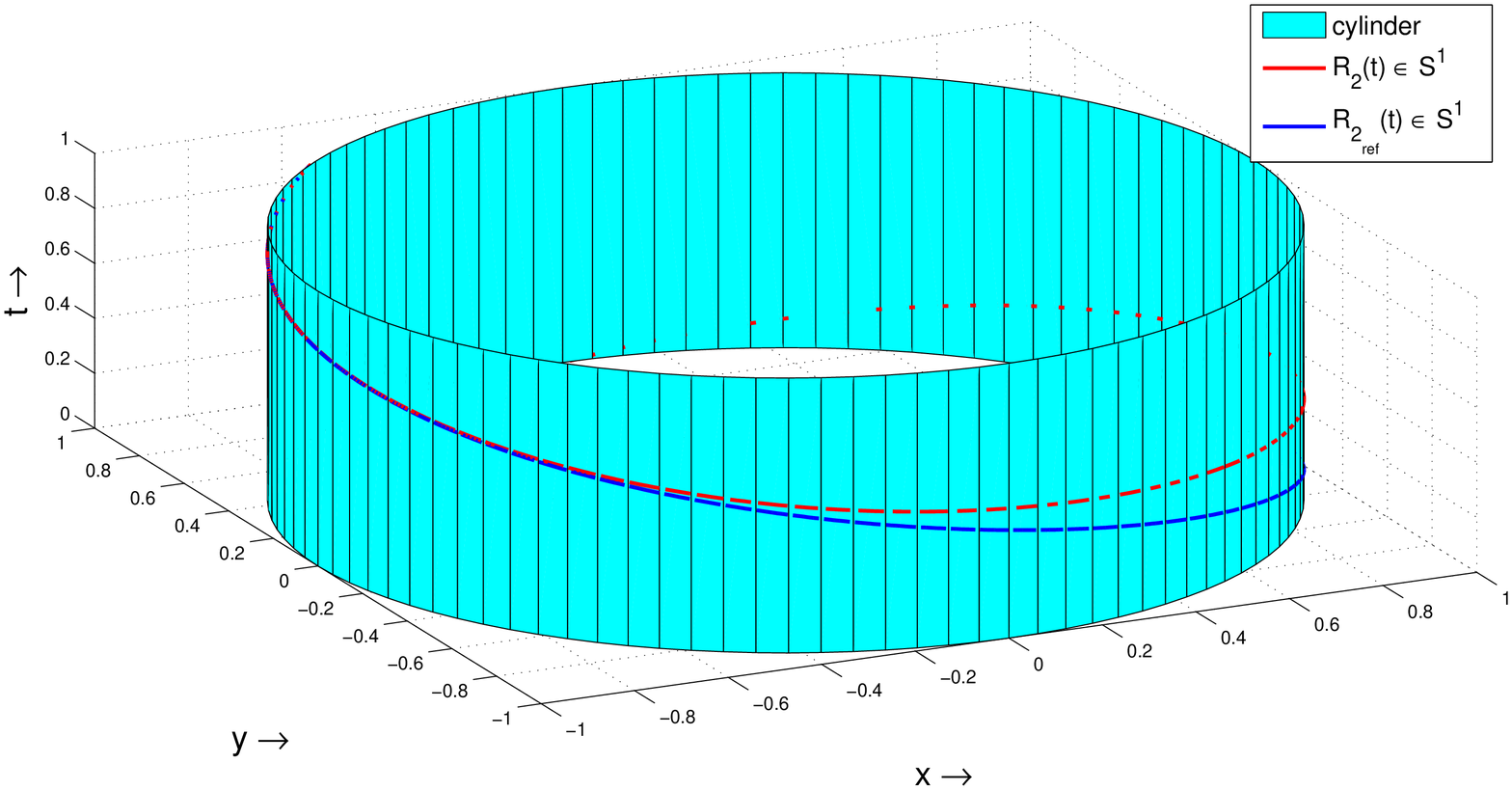}
  \end{subfigure}
 \caption{Tracking results for second set of initial conditions and, $R_{2_{ref}}(t) = \protect \begin{pmatrix} \cos(t) & -\sin(t) \\ \sin(t) & \cos(t)\protect \end{pmatrix}.$}
 \label{fig3ch3}
\end{figure}
\begin{figure}[!h]
\begin{subfigure}[b]{0.15\textwidth}
   %\centering
  \includegraphics[height=2cm, width=\linewidth]  {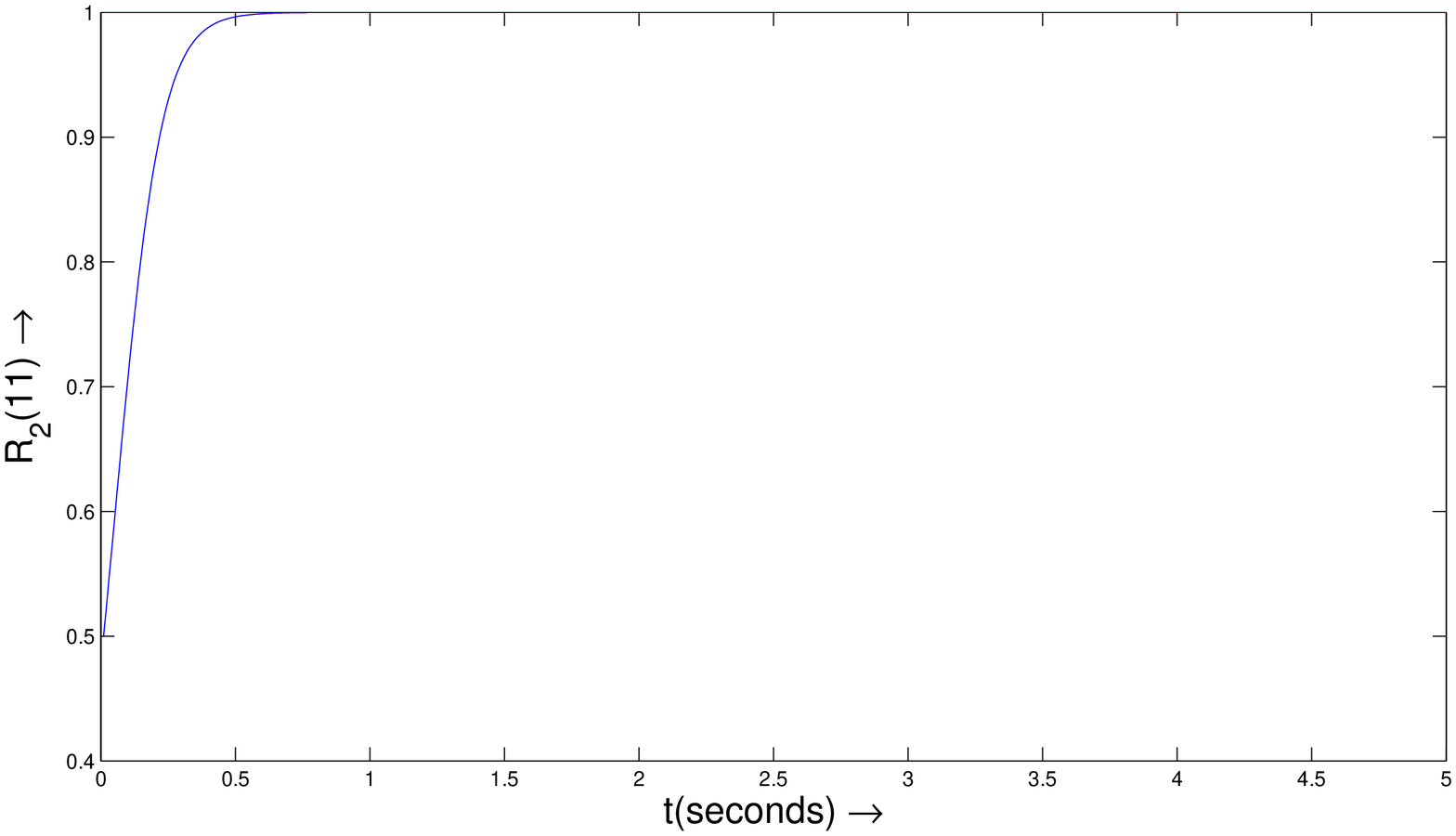}
  \end{subfigure}
  \begin{subfigure}[b]{0.15\textwidth}
   %\centering
  \includegraphics[height=2cm, width=\linewidth]  {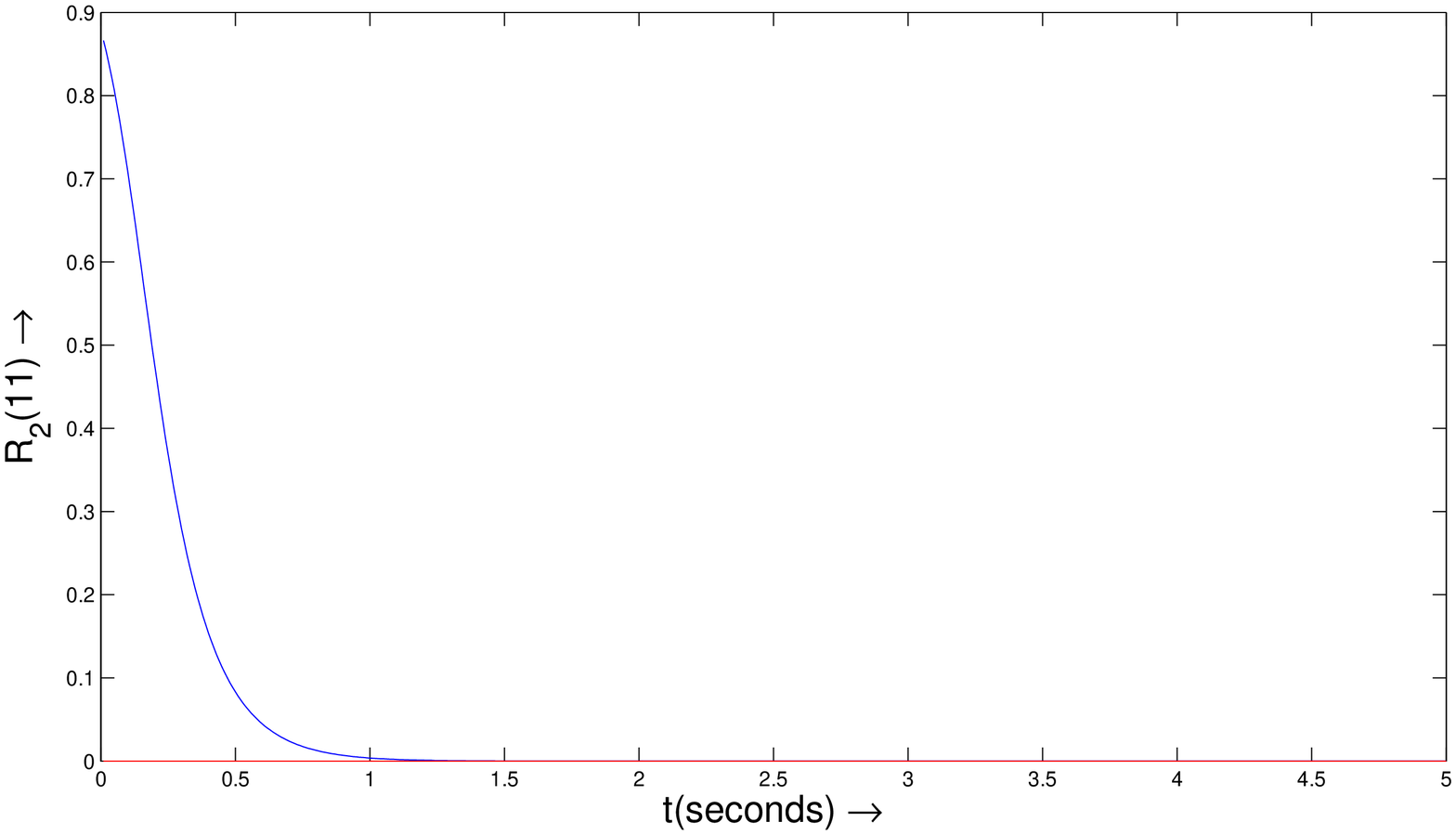}
  \end{subfigure}
  \begin{subfigure}[b]{0.15\textwidth}
   %\centering
  \includegraphics[height=2cm, width=\linewidth]  {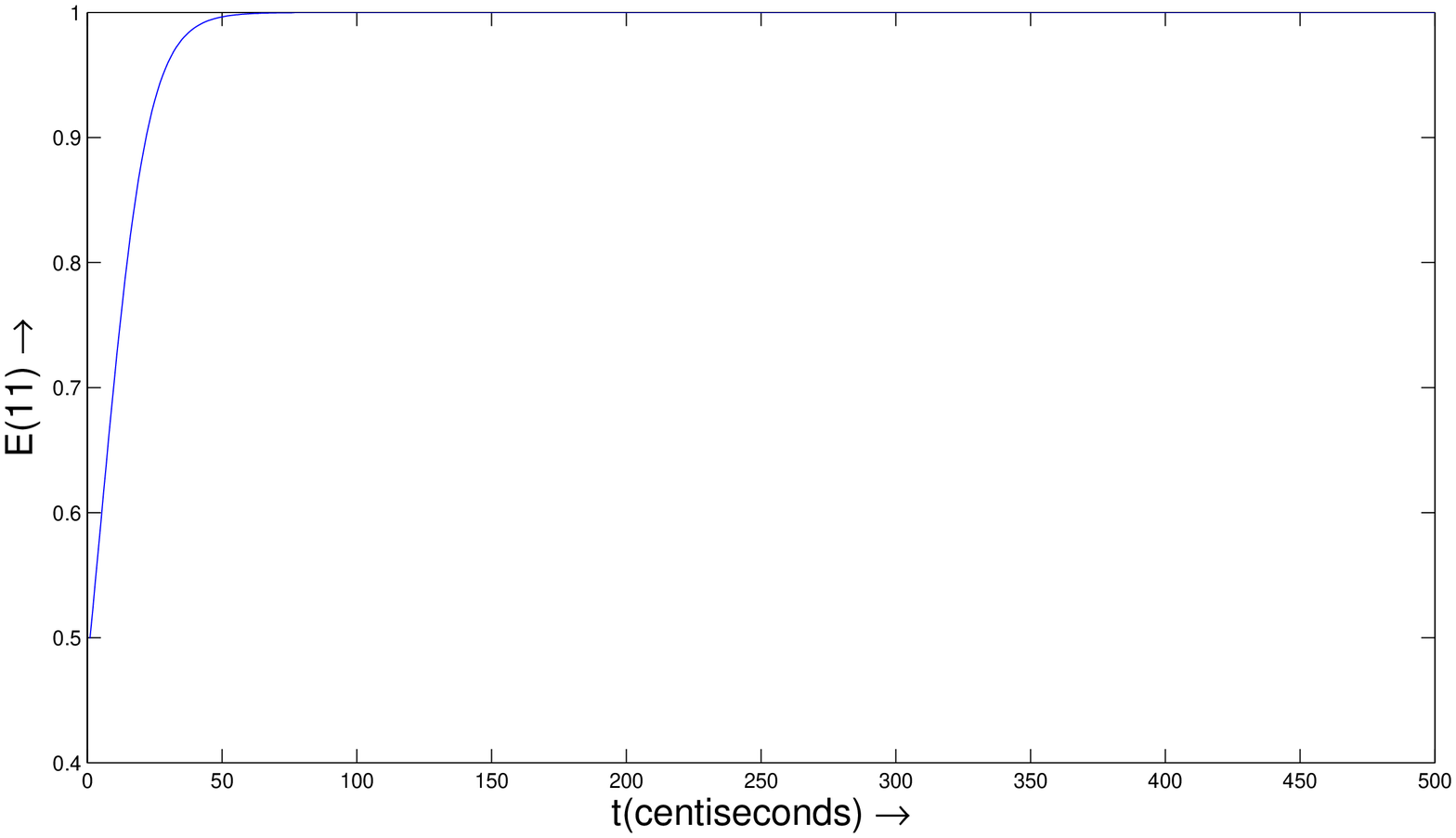}
  \end{subfigure}
  \begin{subfigure}[b]{0.15\textwidth}
   %\centering
  \includegraphics[height=2cm, width=\linewidth]  {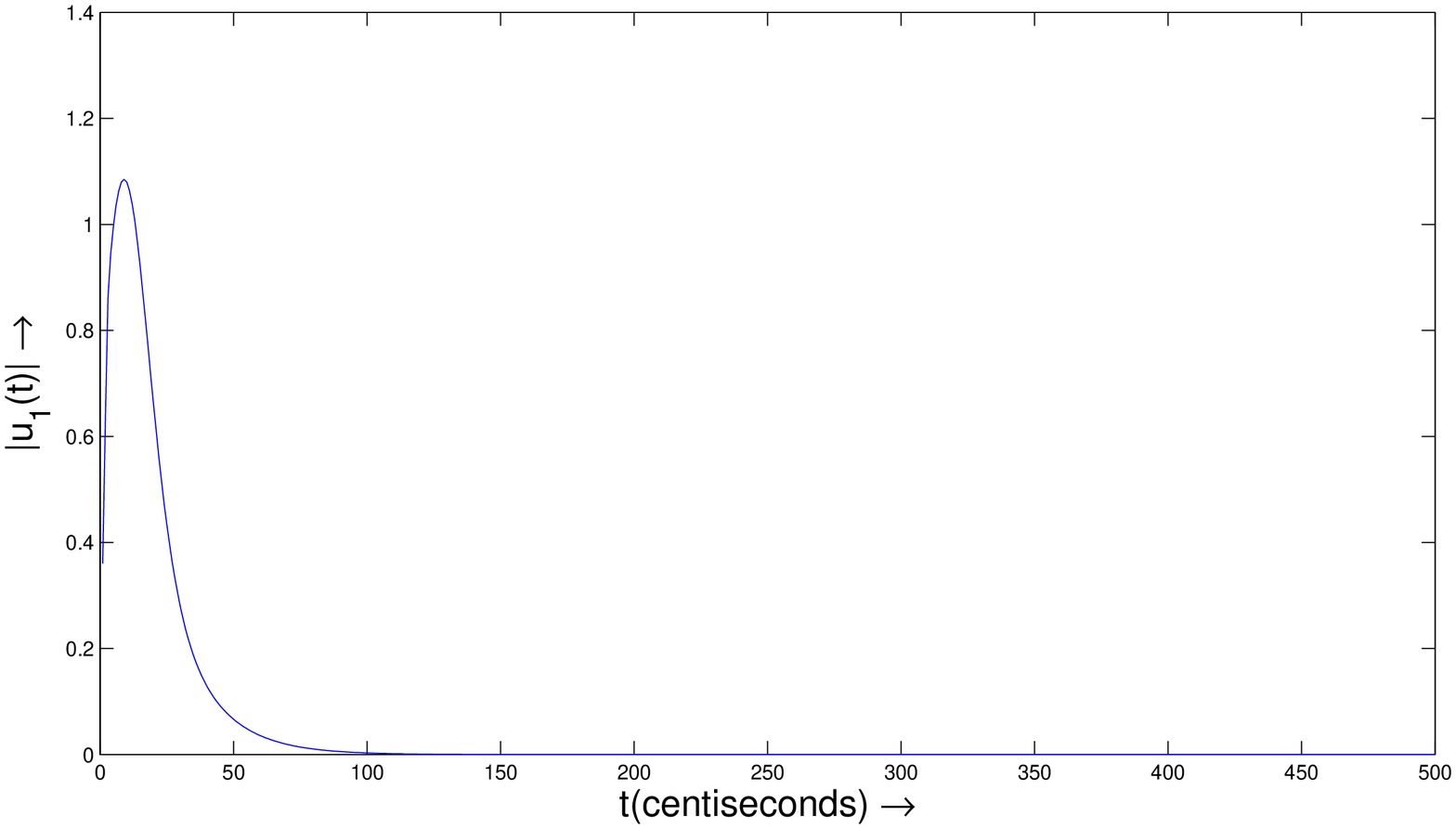}
  \end{subfigure}
  \begin{subfigure}[b]{0.15\textwidth}
   %\centering
  \includegraphics[height=2cm, width=\linewidth]  {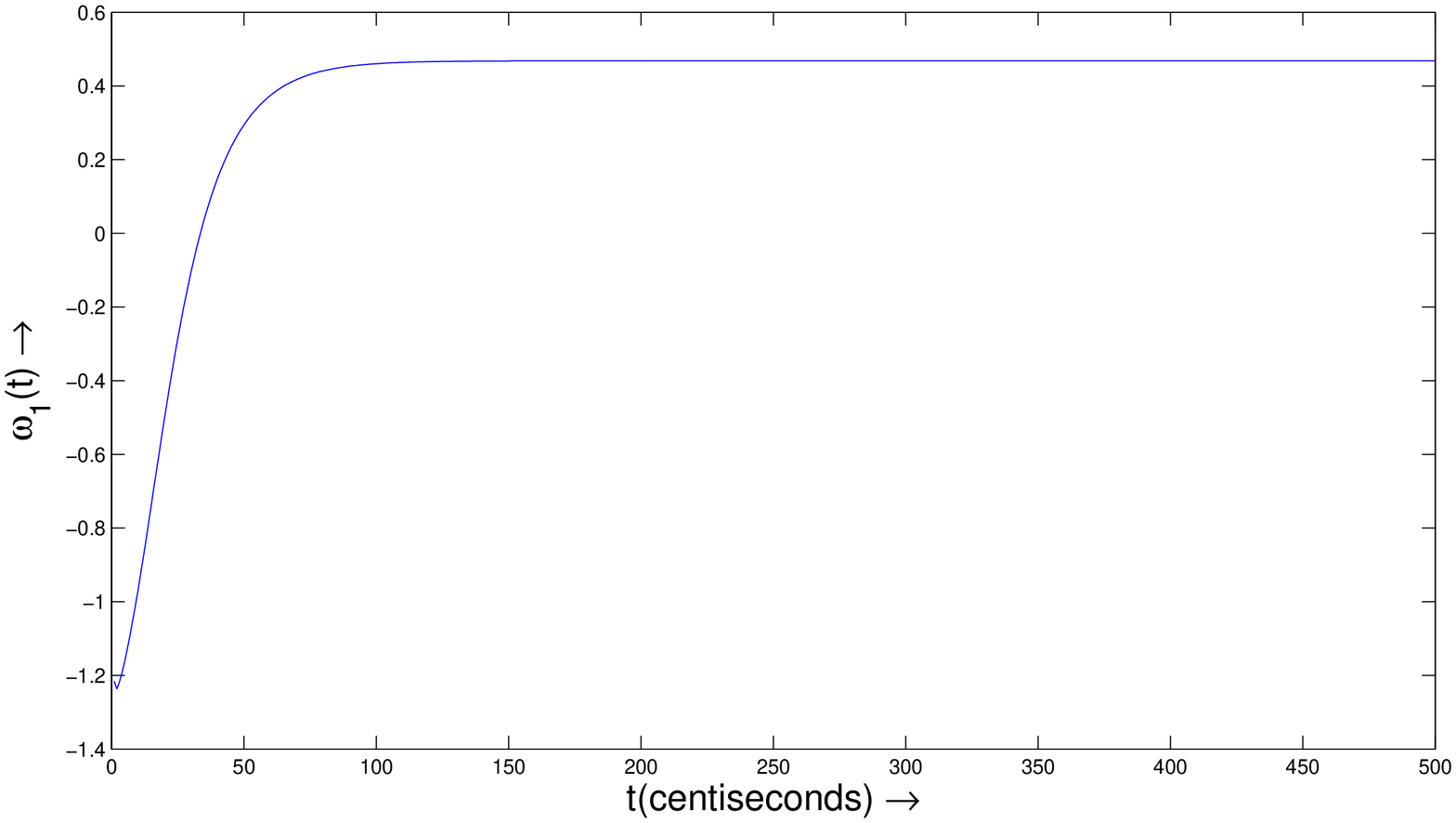}
  \end{subfigure}
 \caption{Tracking results for second set of initial conditions and, $R_{2_{ref}}(t) = id$}
 \label{fig32ch3}
\end{figure}

\begin{equation*}
R_1(0)= \begin{pmatrix} 1 & 0 \\ 0 & 1 \end{pmatrix}, \; \; \omega_1(0) =-1 \quad rad s^{-1}
\end{equation*}
\begin{equation*}
R_2(0)= \begin{pmatrix} 0.4794 & 0.87758 \\ -0.87758 & 0.4794 \end{pmatrix} , \;\; \omega_2(0) =2 \quad rad s^{-1}
\end{equation*}
 The reference trajectory is chosen to be:
 \begin{equation*}
 R_{2_{ref}}(t) = \begin{pmatrix} -\cos(t)& \sin(t) \\ - \sin(t) &  -\cos(t) \end{pmatrix}, \;\; \omega_{2_{ref}}(t) =1 \quad rad s^{-1}, \; t \geq 0.
 \end{equation*}
We consider $K_p= 2.5$, $F_d=- diag(1.5 \quad 2.2)$, $P = diag(2 \quad 1.5)$. The results are shown in Figure ~\ref{fig4ch3}.
\begin{figure}[!h]
\centering
\begin{subfigure}[b]{0.22\textwidth}
  \includegraphics[height=2cm, width=\linewidth]{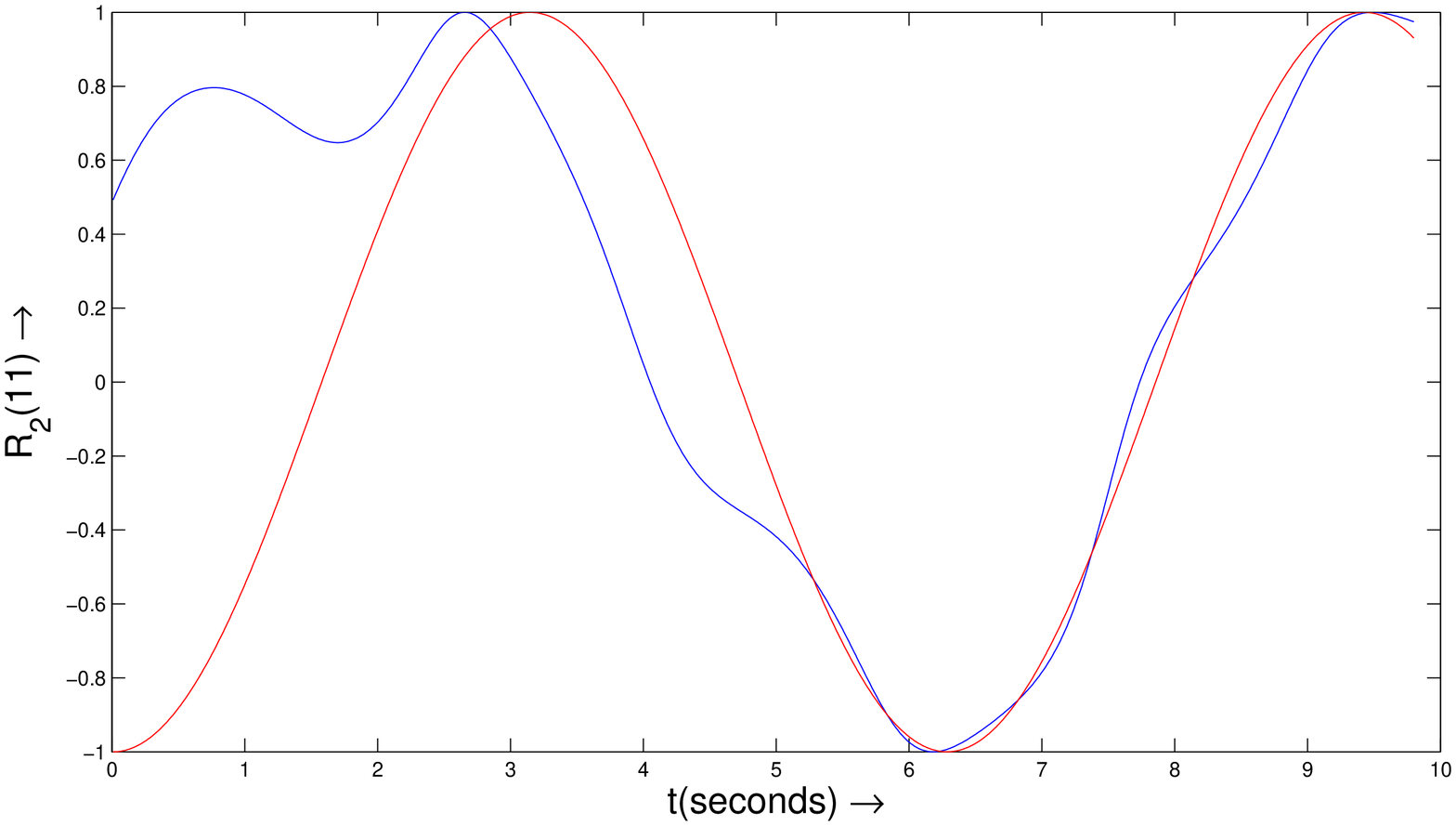}
  \end{subfigure}
\begin{subfigure}[b]{0.22\textwidth}
  \includegraphics[height=2cm, width=\linewidth]{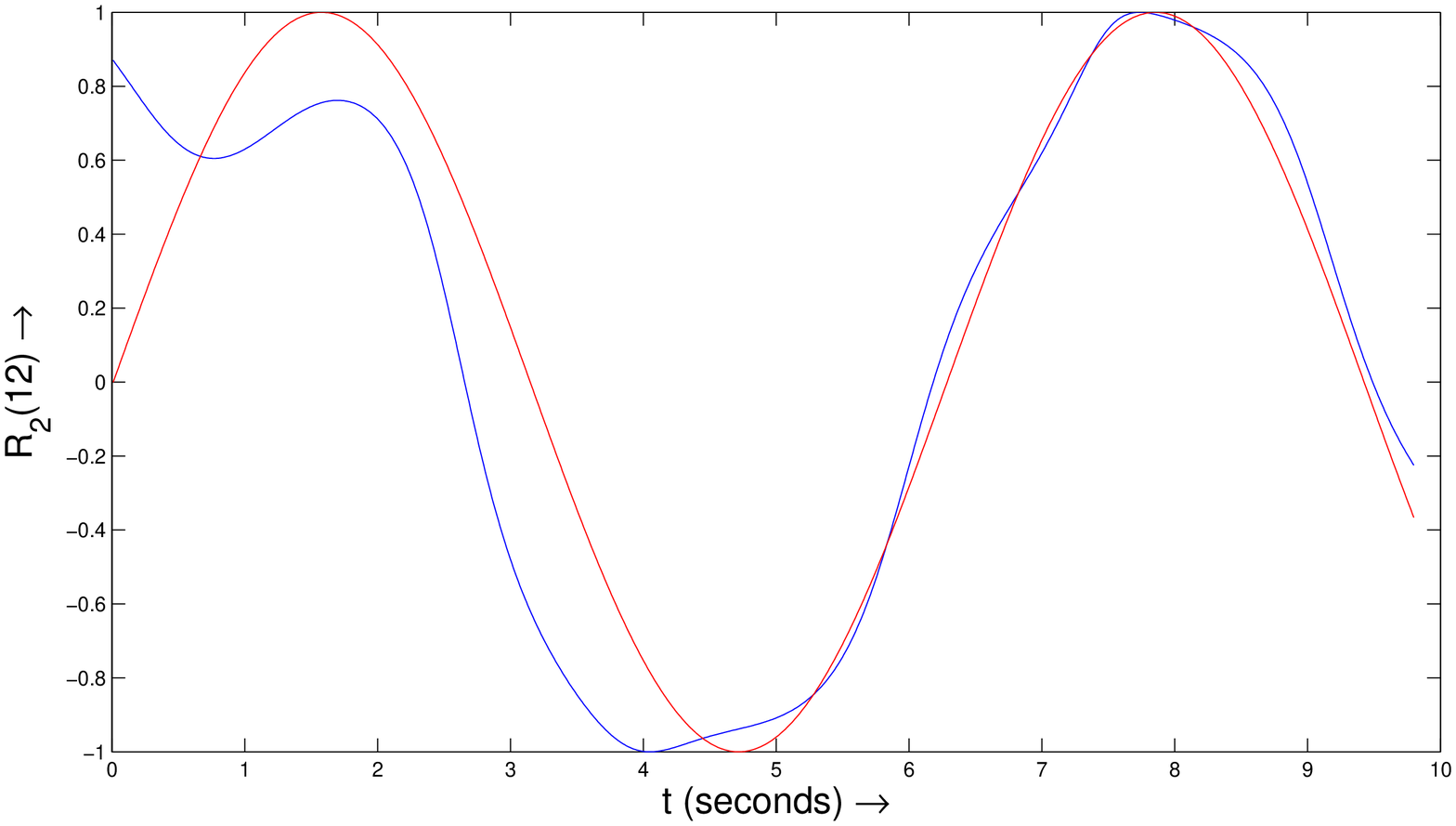}
\end{subfigure}
\begin{subfigure}[b]{0.22\textwidth}
  \includegraphics[height=2cm, width=\linewidth]{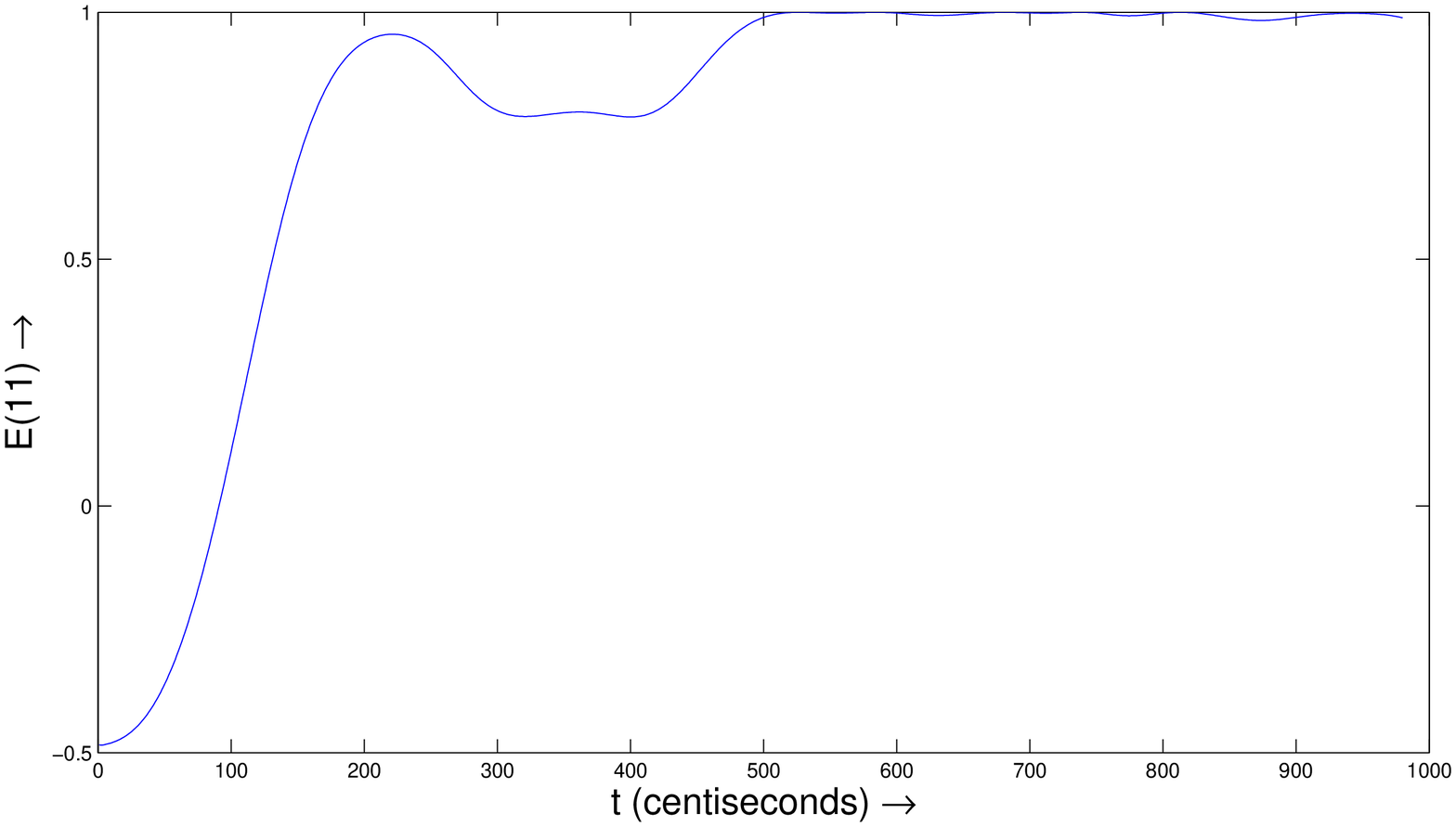}
\end{subfigure}
\begin{subfigure}[b]{0.22\textwidth}
   %\centering
  \includegraphics[height=2cm, width=\linewidth]  {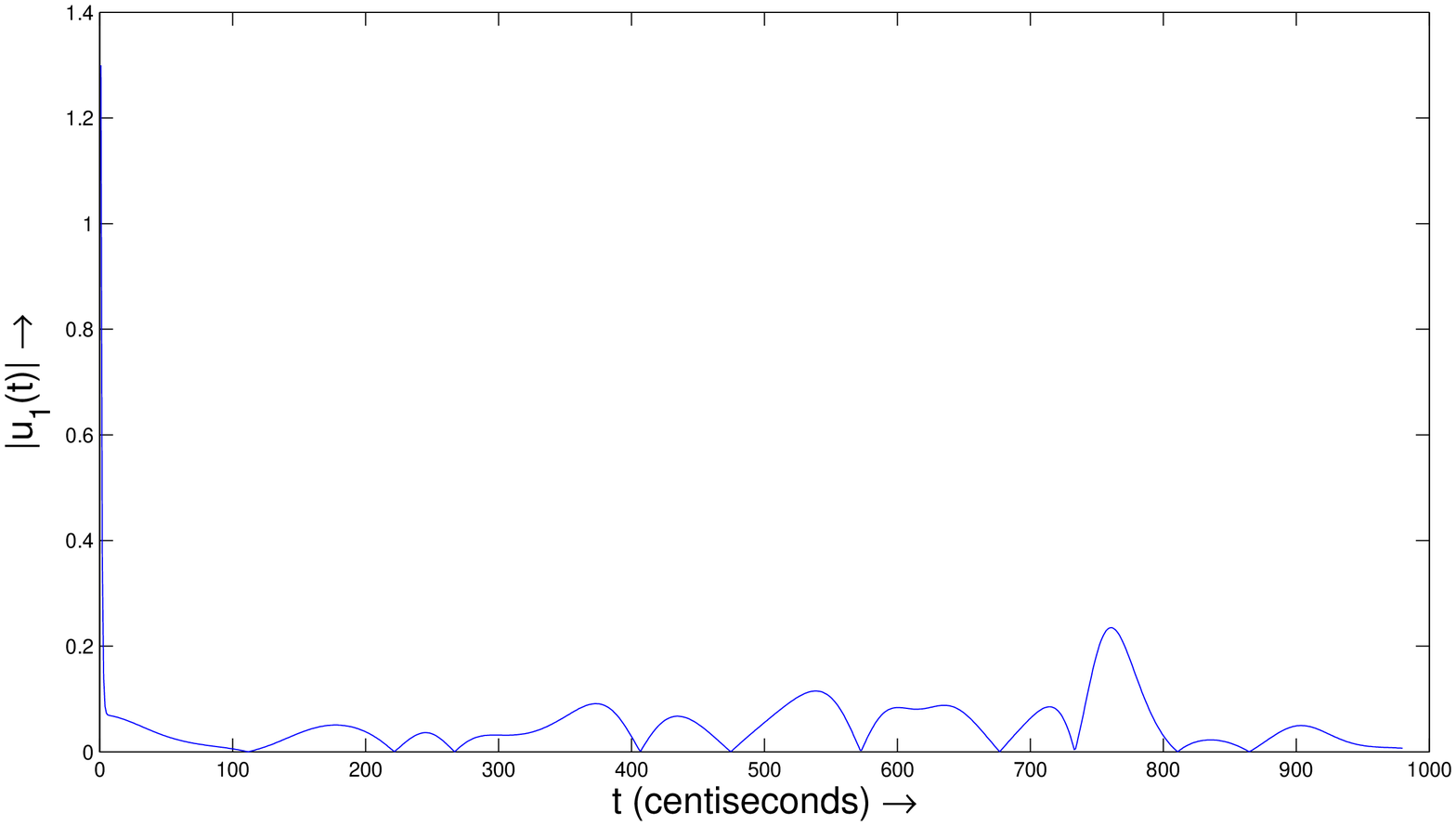}
  \end{subfigure}
  \begin{subfigure}[b]{0.22\textwidth}
   %\centering
  \includegraphics[height=2cm, width=\linewidth]  {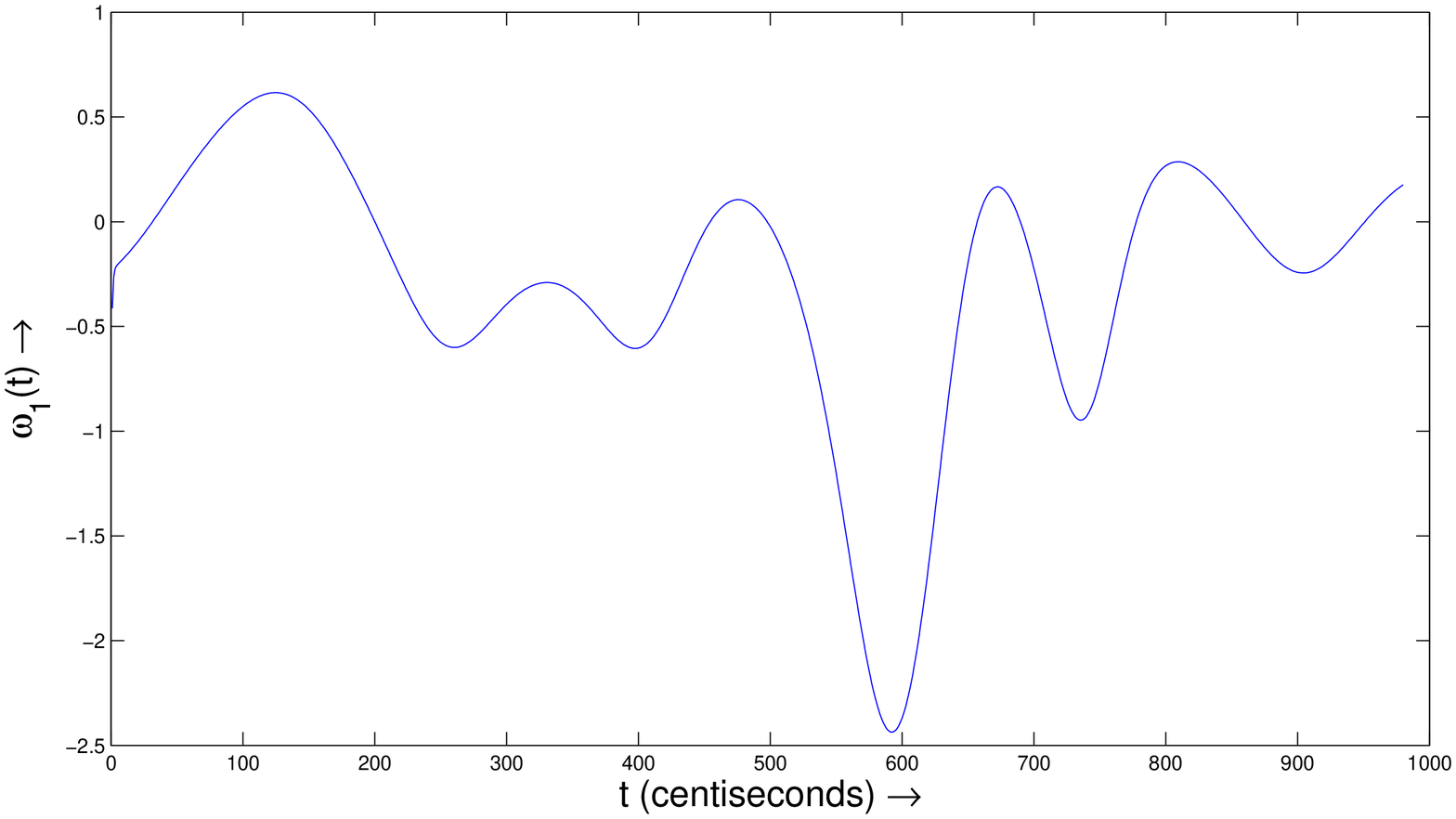}
  \end{subfigure}
 \caption{Tracking results for the last set of initial conditions}
 \label{fig4ch3}
\end{figure}
Next, we consider another pendubot with the following parameters: $m_1=0.1 kg$, $m_2=0.4 kg$, $l_1= 0.25 \quad m$ and $l_2= 0.5 \quad m$. The second set of initial conditions is considered along with:
\begin{equation*}
R_{2_{ref}}(t) = \begin{pmatrix} \cos(t) & -\sin(t) \\ \sin(t) & \cos(t) \end{pmatrix}, \;\;\omega_{2_{ref}}(t) =1 \quad rad s^{-1} ,\; t \geq 0.
\end{equation*}
Simulations are performed with $K_p= 5$, $F_d= -diag(1.5 \quad 2.5)$, $P= diag(1.5 \quad 1.3)$ and results are plotted in Figure \ref{fig5}.
\begin{figure}[!h]
\centering
\begin{subfigure}[b]{0.22\textwidth}
  \includegraphics[height=2cm, width=\linewidth]{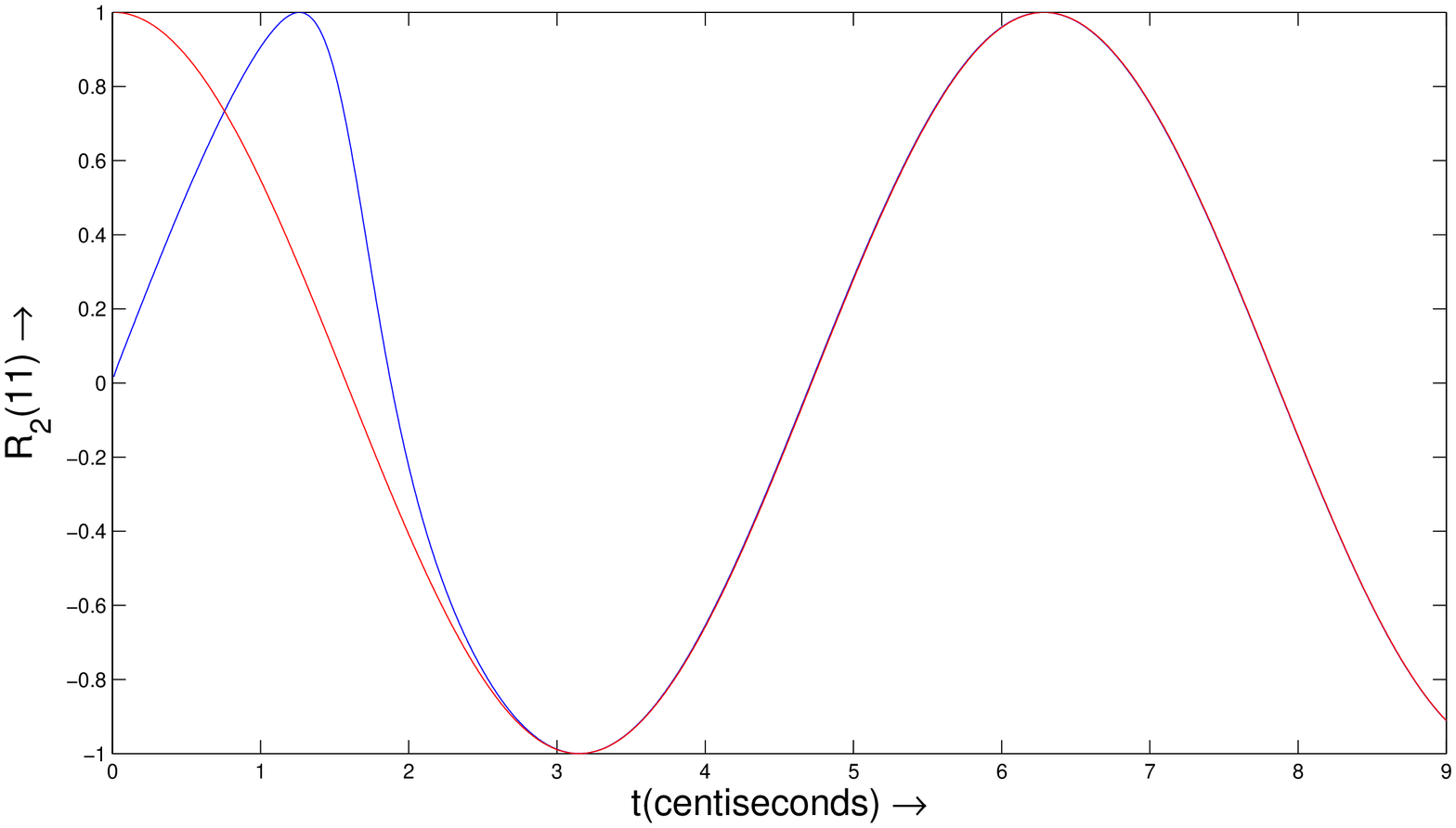}
  \end{subfigure}
\begin{subfigure}[b]{0.22\textwidth}
  \includegraphics[height=2cm, width=\linewidth]{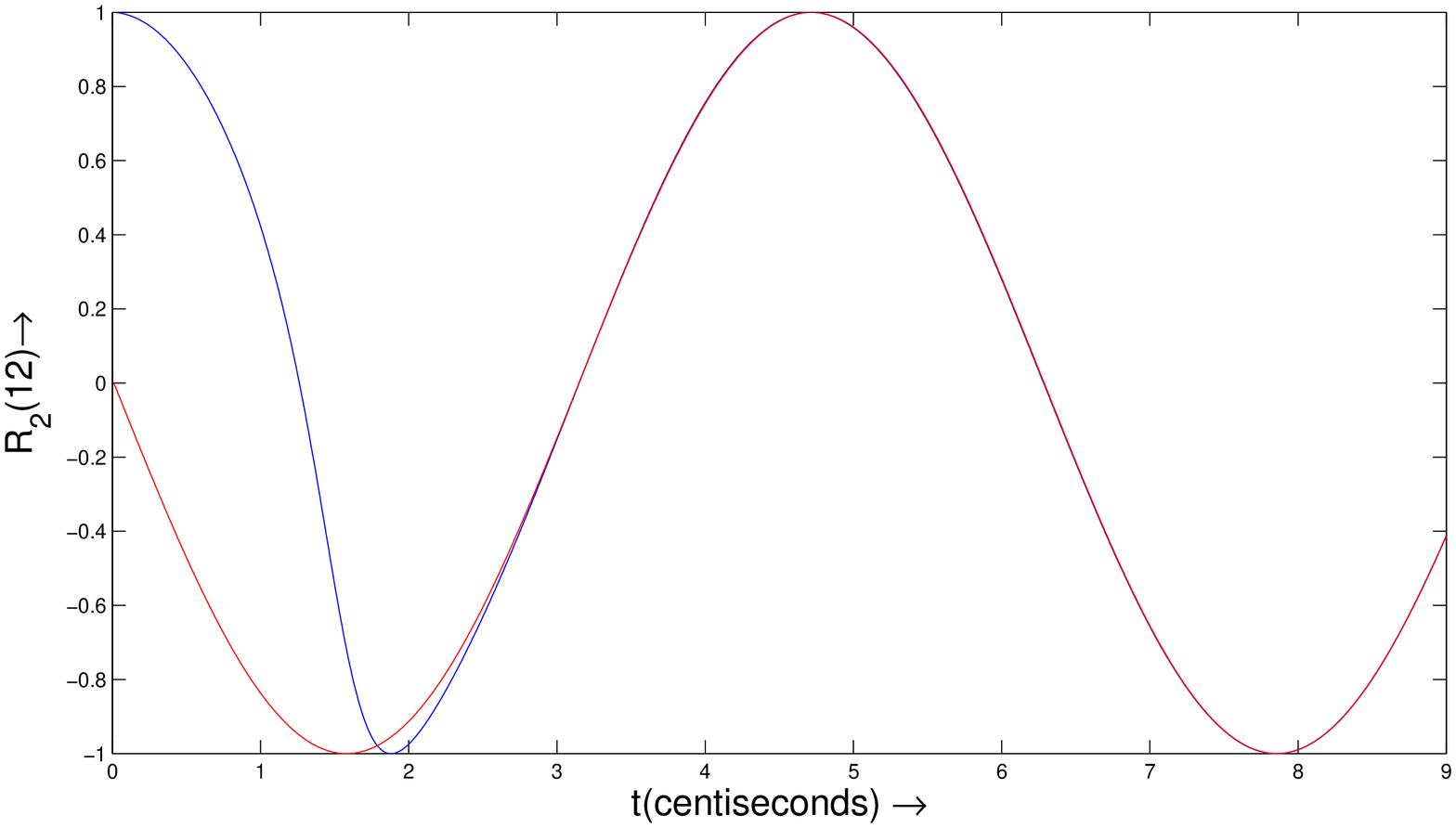}
\end{subfigure}
\begin{subfigure}[b]{0.22\textwidth}
  \includegraphics[height=2cm, width=\linewidth]{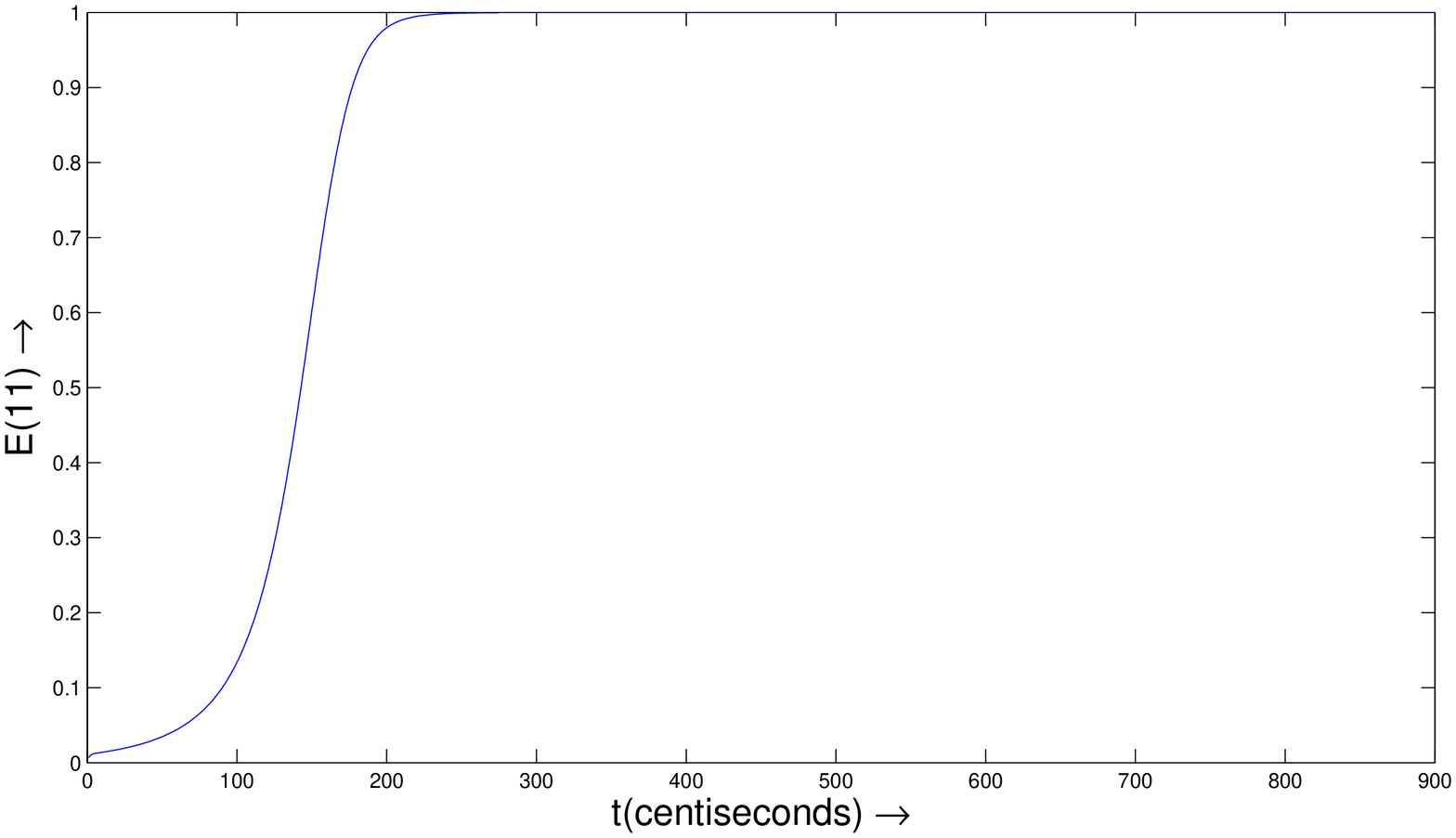}
\end{subfigure}
\begin{subfigure}[b]{0.22\textwidth}
   %\centering
  \includegraphics[height=2cm, width=\linewidth]  {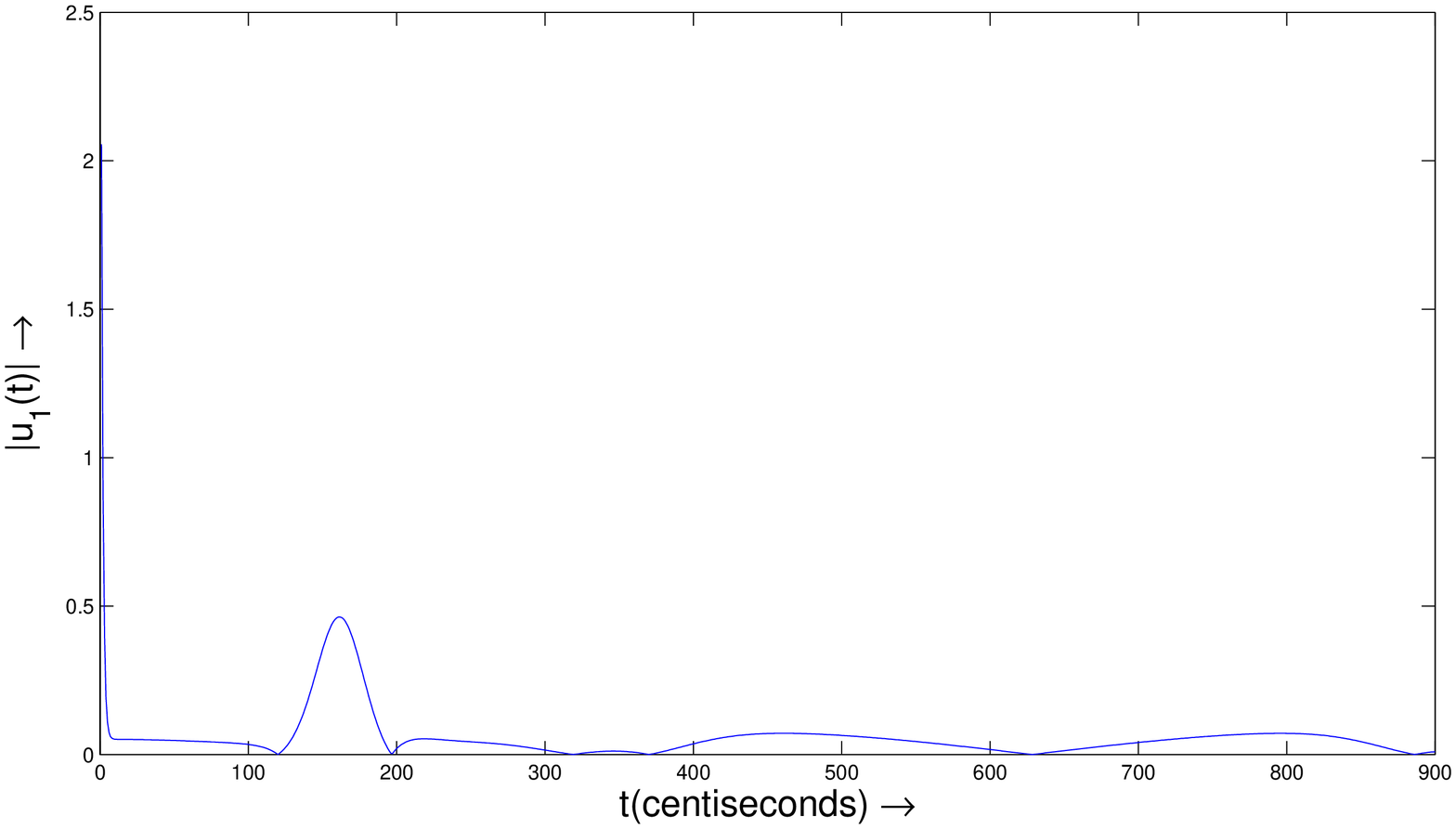}
  \end{subfigure}
  \begin{subfigure}[b]{0.22\textwidth}
   %\centering
  \includegraphics[height=2cm, width=\linewidth]  {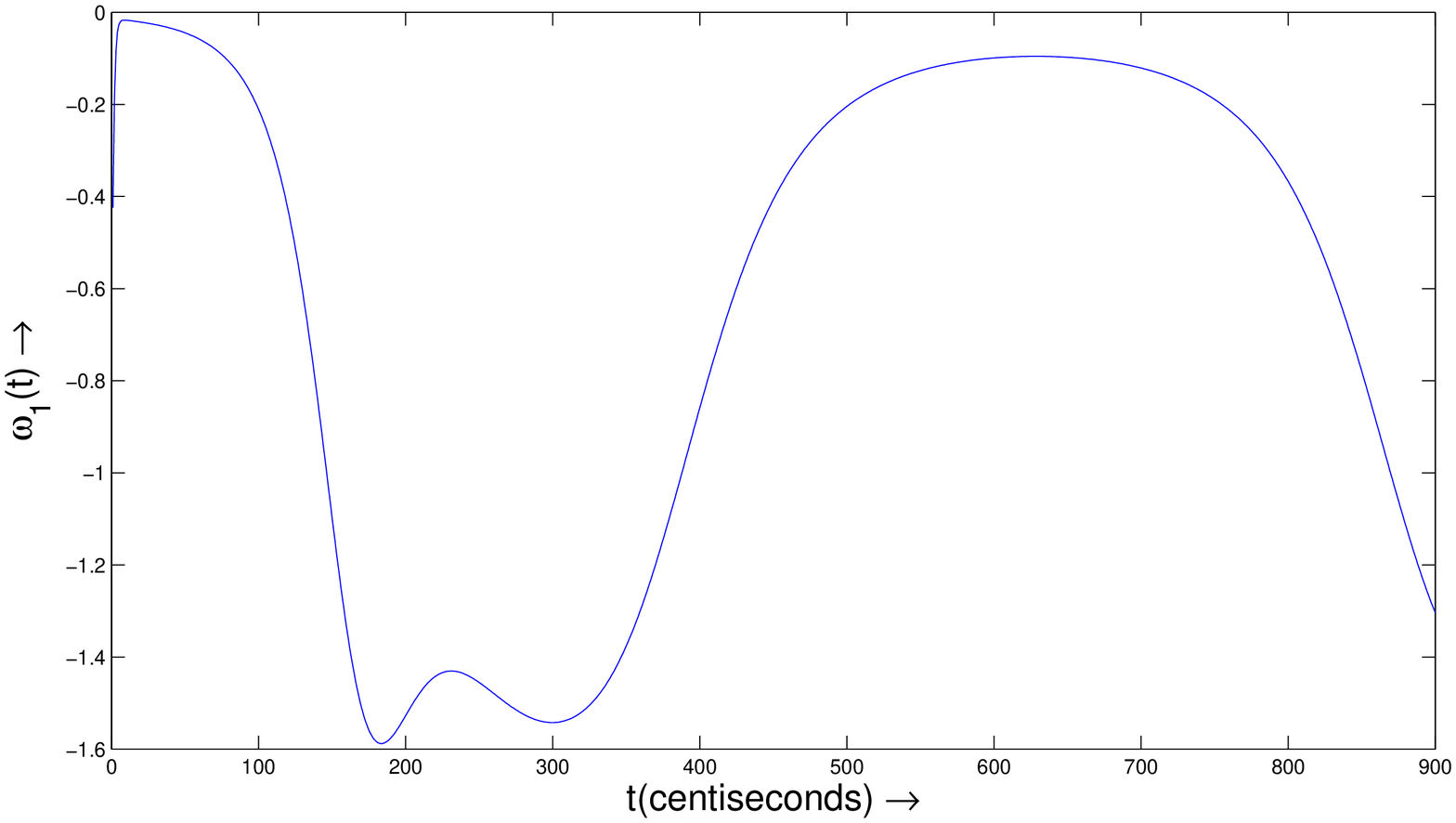}
  \end{subfigure}
 \caption{Tracking results for the second pendubot}
 \label{fig5}
\end{figure}

\section{Conclusions}
By transforming the dynamical equation governing the unactuated coordinate into an SMS through a suitable transformation of the
input, and then adopting standard techniques for tracking control, we synthesize a tracking control  law for the unactuated variable.
A few features of the control law, the simulations for an $n$-link manipulator and drawing conclusions based on the
performance, constitute some ongoing work.

\appendix
\subsection{Kinetic energy}\label{app1}
In what follows, we derive the expression for kinetic energy of the pendubot in \eqref{ke1}.
\begin{align}\label{xbidot}
 \dot{x}_{B_1}  &= \dot{R}_1 X_{B_1} \\ \nonumber
  \dot{x}_{B_2} &= \dot{R}_1 L_1 +( \dot{R}_1 R_2 + {R}_1\dot{R}_2) X_{B_2}
\end{align}
Therefore,
\begin{align}\label{normxb1}
 ||\dot{x}_{B_1}||^2 &= \langle \dot{R}_1 X_{B_1}, \dot{R}_1 X_{B_1} \rangle = \langle \hat{\omega}_1 X_{B_1}, \hat{\omega}_1 X_{B_1} \rangle \\ \nonumber
 &= (\hat{\omega}_1 X_{B_1})^T (\hat{\omega}_1 X_{B_1})= \Omega^T (\hat{\omega}_1 X_{B_1} \times  X_{B_1}).
\end{align}
where $\Omega= \begin{pmatrix}
                 0 & 0 & \omega
               \end{pmatrix}^T$. From \eqref{normxb1}, the first term in \eqref{ke} is
\begin{equation}\label{t1ke}
  \int_{B_1} ||\dot{x}_{B_1}||^2 \rho_1 \mathrm{d}V_1= \int_{B_1} \Omega^T (\hat{\omega}_1 X_{B_1} \times  X_{B_1}) \rho_1 \mathrm{d}V_1 = \mathbb{I}_1 \omega^2_1.
\end{equation}
From \eqref{xbidot},
\begin{align}\label{normxb2}
||\dot{x}_{B_2}||^2 &= \langle \dot{R}_1 L_1,\dot{R}_1 L_1 \rangle + 2 \langle  \dot{R}_1 L_1,(\dot{R}_1 R_2+ {R}_1\dot{R}_2)X_{B_2} \rangle \\ \nonumber
&+ \langle(\dot{R}_1 R_2+ {R}_1\dot{R}_2)X_{B_2},(\dot{R}_1 R_2+ {R}_1\dot{R}_2)X_{B_2} \rangle\\ \nonumber
&= \langle \hat{\omega}_1L_1, \hat{\omega}_1 L_1 \rangle + 2 \langle \hat{\omega}_1 L_1 ,(\hat{\omega}_1 R_2 + R_2 \hat{\omega}_2 )X_{B_2}\rangle\\ \nonumber
&+ \langle (\hat{\omega}_1R_2+ R_2 \hat{\omega}_2)X_{B_2}, (\hat{\omega}_1R_2+ R_2 \hat{\omega}_2)X_{B_2} \rangle \\\nonumber
&= \langle \hat{\omega}_1L_1, \hat{\omega}_1 L_1 \rangle + 2 (\langle \hat{\omega}_1 L_1 , \hat{\omega}_1 (R_2 X_{B_2}) \rangle \\ \nonumber
&+ \langle \hat{\omega}_1 L_1 , R_2\hat{\omega}_2 L_2 \rangle)\\ \nonumber
&+ \langle (\hat{\omega}_1R_2+ R_2 \hat{\omega}_2)X_{B_2}, (\hat{\omega}_1R_2+ R_2 \hat{\omega}_2)X_{B_2} \rangle
\end{align}
We expand the second term in \eqref{normxb2},
\begin{align*}
&T_2 \coloneq 2 \langle \hat{\omega}_1 L_1 , \hat{\omega}_1 (R_2 X_{B_2}) \rangle + 2\langle \hat{\omega}_1 L_1 , R_2\hat{\omega}_2 L_2 \rangle\\
&= 2\langle \Omega_1, (R_2 X_{B_2}) \times (\hat{\omega}_1 L_1 )\rangle \\
&+ 2\langle \Omega_1, L_1 \times R_2(\hat{\omega}_2X_{B_2})\rangle\\
&= 2\langle \Omega_1, (L_1^T R_2 X_{B_2}) \Omega_1 - ((R_2 X_{B_2})^T\Omega_1)L_1 \rangle\\
&+ 2 \langle \Omega_1, (L_1^T R_2 X_{B_2})(R_2 \Omega_2)- (L_1^T R_2\Omega_2)R_2 X_{B_2} \rangle \\
&= 2\langle \Omega_1, (L_1^T R_2 X_{B_2})(\Omega_1+\Omega_2) \rangle
\end{align*}
as $(R_2 X_{B_2})^T\Omega_1 =0$, $R_2\Omega_2 = \Omega_2$ and, $L_1^T \Omega_2=0$.\\
The third term in \eqref{normxb2} is
\begin{align*}
&T_3 \coloneq \langle (\hat{\omega}_1R_2+ R_2 \hat{\omega}_2)X_{B_2}, (\hat{\omega}_1R_2+ R_2 \hat{\omega}_2)X_{B_2} \rangle\\
&= \langle \hat{\omega}_1R_2 X_{B_2}, \hat{\omega}_1R_2 X_{B_2} \rangle
+2\langle R_2 \hat{\omega}_2 X_{B_2}, \hat{\omega}_1R_2 X_{B_2} \rangle\\
&+ \langle R_2 \hat{\omega}_2 X_{B_2},R_2 \hat{\omega}_2 X_{B_2}\rangle\\
&= \langle \hat{\omega}_1 (R_2 X_{B_2} \times \hat{\omega}_1 R_2 X_{B_2}) \rangle + \langle \hat{\omega}_1 ,R_2 (X_{B_2} \times \hat{\omega_2}X_{B_2})\\
&+ \langle \hat{\omega}_2 X_{B_2}, \hat{\omega}_2 X_{B_2}\rangle
\end{align*}
Observe that,
\begin{equation*}
\int_{B_2} T_2 \rho_2 \mathrm{d}V_2 = m_2(L_1^T R_2 L_2)(\omega_1^2+ \omega_1 \omega_2)
\end{equation*}
as $\int_{B_2} X_{B_2} \rho_2 \mathrm{d}V_2 = L_2/2$ and,
\begin{equation*}
\int_{B_2} T_3 \rho_2 \mathrm{d}V_2= \mathbb{I}_2 \omega_1^2+ 2 \mathbb{I}_2 \omega_1 \omega_2 + \mathbb{I}_2 \omega_2^2
\end{equation*}
as $\mathbb{I}_2 (v) = \int_{B_2} X_{B_2} \times (\hat{v} X_{B_2})\rho_2 \mathrm{d}V_2 $ for $v \in \mathbb{R}^3$.
From \eqref{normxb2}, the second term in \eqref{ke} is
\begin{align}\label{t2ke}
  \int_{B_2} ||\dot{x}_{B_2}||^2 \rho_2 \mathrm{d}V_2 &=  \int_{B_2}\langle \hat{\omega}_1L_1, \hat{\omega}_1 L_1 \rangle \rho_2 \mathrm{d}V_2 \\ \nonumber
  &+ \int_{B_2} T_2 \rho_2 \mathrm{d}V_2  + \int_{B_2} T_3 \rho_2 \mathrm{d}V_2\\\nonumber
  &= m_2 l_1^2 \omega_1^2 + m_2(L_1^T R_2 L_2)(\omega_1^2+ \omega_1 \omega_2) \\ \nonumber
  &+ \mathbb{I}_2 \omega_1^2+ 2 \mathbb{I}_2 \omega_1 \omega_2 + \mathbb{I}_2 \omega_2^2
\end{align}
Therefore, from \eqref{ke}, \eqref{t1ke} and \eqref{t2ke}, the kinetic energy is
\begin{align} \label{ke2}
K(\omega_1, \omega_2, R_2) = &\omega_1^2(\mathbb{I}_1 + \mathbb{I}_2 + m_2 l_1^2+m_2L_1^T R_2 L_2 ) \\ \nonumber
&+ \omega_1 \omega_2 ( 2 \mathbb{I}_2+ m_2L_1^T R_2 L_2 )+ \mathbb{I}_2 \omega_2^2
\end{align}
\subsection{Equations of motion} \label{app2}
The Lagrangian is defined as $L = K-V$ where $K$ is given in \eqref{ke1} and $V$ is given in \eqref{pe}. As the Lagrangian is not invariant with respect to $R_2$, therefore the equations of motion have to be derived using method of variations. Let $R_1(t)$, $R_2(t)$, $t \in [a,b]$ be curves on $SO(2)$ with $\delta R_i(a) = R_i(b)=0$ for $i =1,2$. By Hamilton's principle, the variation of the action integral is zero. Therefore,
\begin{equation}
\delta \int_a^b L(R_1,R_2,\omega_1,\omega_2)=0
\end{equation}
which is
\begin{align}\label{varprin}
&\int_a^b \{ \langle \frac{\delta K}{\partial \omega_1}, \delta \omega_1 \rangle +  \langle \frac{\delta K}{\partial \omega_2}, \delta \omega_2 \rangle + \langle \frac{\delta K}{\partial R_2}, \delta R_2 \rangle  \\ \nonumber
&+ \langle \frac{\delta V}{\partial R_1}, \delta R_1 \rangle  + \langle \frac{\delta V}{\partial R_1}, \delta R_1 \rangle\}\mathrm{d}t =0
\end{align}
where the variations $\delta \omega_1$, $\delta \omega_2$ are induced by the variations in $R_1(t)$ and $R_2(t)$ respectively as follows
\begin{equation}\label{domega}
\delta \hat{\omega}_i = - R_i^T \delta R_i R_i^T \dot{R}_i + R_i^T \delta \dot{R}_i, \quad \text{for} \quad i=1,2.
\end{equation}
Let, $\hat{\Sigma}_i = R^T_i \delta R_i $, $i =1,2$. As $SO(2)$ is Abelian,
\begin{equation}\label{dsigma}
\dot\hat{\Sigma}_i  = \delta \hat{\omega}_i ,  \quad \text{for} \quad i=1,2.
\end{equation}
Therefore, $\dot{\Sigma}_i = \delta {\omega}_i  $, $i=1,2$.
\newline
Let $K = \omega_1^2 K_1 + \omega_1 \omega_2 K_2 + \omega_2^2 K_3$ and let $V = T_1 +T_2$ where $T_1 = (\frac{m_1}{2}+m_2)g l_1 e_1^T R_1 e_1 $ and $T_2 = \frac{m_2}{2}g l_2 e_1^T R_2 R_1 e_1$. \\
From \eqref{ke1}, $K_1 = \mathbb{I}_1 + \mathbb{I}_2 + m_2 l_1^2+m_2L_1^T R_2 L_2 $, $K_2 =2 \mathbb{I}_2+ m_2L_1^T R_2 L_2 $ and, $K_3 = \mathbb{I}_2$. Now we expand each term of \eqref{varprin},
\begin{equation}\label{t1}
\langle \frac{\delta K}{\partial \omega_1}, \delta \omega_1 \rangle = \langle 2 K_1 \omega_1 + K_2 \omega_2, \dot{\Sigma}_1 \rangle,
\end{equation}
\begin{equation}\label{t2}
\langle \frac{\delta K}{\partial \omega_2}, \delta \omega_2 \rangle= \langle 2 K_3 \omega_2 + K_2 \omega_1 , \dot{\Sigma}_2 \rangle
\end{equation}
\begin{align}\label{t3}
\langle \frac{\delta K}{\partial R_2}, \delta R_2 \rangle &= \langle m_2 L_1 L_2^T(\omega_1^2 + \omega_1 \omega_2), \delta R_2 \rangle \\ \nonumber
&= \langle m_2 R_2 L_1 L_2^T(\omega_1^2 + \omega_1 \omega_2), \delta \hat{\Sigma}_2 \rangle
\end{align}
\begin{equation}\label{t4}
\langle \frac{\delta V}{\partial R_1}, \delta R_1 \rangle= \langle R_1(\frac{\delta T_1}{\partial R_1} + \frac{\delta T_2}{\partial R_1}), \delta \hat{\Sigma}_1 \rangle
\end{equation}
\begin{equation}\label{t5}
\langle \frac{\delta V}{\partial R_2}, \delta R_2 \rangle= \langle R_2\frac{\delta T_2}{\partial R_2}, \delta \hat{\Sigma}_2 \rangle
\end{equation}
From \eqref{t1}, integrating by parts we get,
\begin{equation}\label{t1cont}
\int_a^b \langle \frac{\delta K}{\partial \omega_1}, \delta \omega_1 \rangle = \langle - \frac{\mathrm{d}}{\mathrm{d}t}(2K_1 \omega_1 +K_2 \omega_2), \Sigma_1 \rangle
\end{equation}
as $\Sigma_1(a) = \Sigma_1(b) = 0$. Similarly, from \eqref{t2}, integrating by parts we get,
\begin{equation}\label{t2cont}
\int_a^b\langle \frac{\delta K}{\partial \omega_2}, \delta \omega_2 \rangle=  \langle - \frac{\mathrm{d}}{\mathrm{d}t}(2K_3 \omega_2 +K_2 \omega_1), \Sigma_2 \rangle
\end{equation}
Therefore, from \eqref{t3}, \eqref{t1cont} and \eqref{t2cont}, \eqref{varprin} is,
\begin{align}\label{varprin1}
\int_a^b  \langle &- {\bigg\{\frac{\mathrm{d}}{\mathrm{d}t}(2K_1 \omega_1 +K_2 \omega_2)\bigg\}}^{\hat{}} -R_1(\frac{\delta T_1}{\partial R_1} + \frac{\delta T_2}{\partial R_1}), \hat{\Sigma}_1 \rangle \\ \nonumber
&+ \langle m_2 R_2 L_1 L_2^T(\omega_1^2 + \omega_1 \omega_2) -R_2\frac{\delta T_2}{\partial R_2}\\ \nonumber
&-  {\bigg\{\frac{\mathrm{d}}{\mathrm{d}t}(2K_3 \omega_2 +K_2 \omega_1)\bigg\}}^{\hat{}}, \hat{\Sigma}_2 \rangle= 0.
\end{align}
Since \eqref{varprin1} holds for all $\Sigma_i$ in definition \eqref{dsigma},
\begin{subequations}
\begin{equation} \label{dyneq1}
{\bigg\{\frac{\mathrm{d}}{\mathrm{d}t}(2K_1 \omega_1 +K_2 \omega_2)\bigg\}}^{\hat{}} + R_1(\frac{\delta T_1}{\partial R_1} + \frac{\delta T_2}{\partial R_1})= u_1 \quad \text{and},
\end{equation}
\begin{align} \label{dyneq2}
 &{\bigg\{\frac{\mathrm{d}}{\mathrm{d}t}(2K_3 \omega_2 +K_2 \omega_1)\bigg\}}^{\hat{}} - m_2( R_2 L_1 L_2^T)(\omega_1^2 + \omega_1 \omega_2) \\ \nonumber
 &+R_2\frac{\delta T_2}{\partial R_2}=0.
\end{align}
\end{subequations}
where $u_1$ is the actuation in $L_1$.
\newline
Note that $L_1^T R_2 L_2  = tr(L_1 L_2^T R_2)$. Therefore,
\begin{equation*}
\frac{\mathrm{d}}{\mathrm{d}t}K_1 = m_2 \langle L_1 L_2^T, R_2 \hat{\omega}_2 \rangle \quad \text{and}, \quad
\frac{\mathrm{d}}{\mathrm{d}t}K_2= m_2 \langle L_1 L_2^T, R_2 \hat{\omega}_2 \rangle .
\end{equation*}
Also,
\begin{align*}
&\frac{\delta T_1}{\partial R_1}= (\frac{m_1}{2}+m_2)g l_1 e_1 e_1^T \quad, \quad \frac{\delta T_2}{\partial R_1} = \frac{m_2}{2} g l_2 R_2^T e_1 e_1^T\\
&\text{and} \quad \frac{\delta T_2}{\partial R_2}= \frac{m_2}{2} g l_2 e_1 e_1^T R_1^T.
\end{align*}
Let, \begin{align*}
\alpha \coloneq m_2 \langle L_1 L_2^T, R_2 \hat{\omega}_2 \rangle, \quad \beta \coloneq m_2{\{skew(R_2 L_1 L_2^T)\}}\breve{},
\end{align*}
\begin{align*}
\Gamma_1 &\coloneq (0.5{m_1}+m_2)g l_1 {\{skew(R_1 e_1 e_1^T)\}}\breve{}\\
&+ 0.5{m_2} g l_2 {\{skew(R_1 R_2^T e_1 e_1^T)\}}\breve{} ,
\end{align*}
\begin{align*}
\Gamma_2 \coloneq 0.5{m_2}g l_2 {\{skew(R_2 e_1 e_1^T R_1^T)\}}\breve{}.
\end{align*}
Therefore, \eqref{dyneq1} and \eqref{dyneq2} are,
\begin{subequations}
\begin{equation} \label{dyneq12}
2K_1 \dot{\omega}_1 +  \alpha (2\omega_1+ \omega_2) + K_2 \dot{\omega}_2 + \Gamma_1=u_1 \quad \text{and},
\end{equation}
\begin{equation}\label{dyneq22}
2K_3 {\dot{\omega}}_2+ \alpha {\omega}_1 + K_2 {\dot{\omega}}_1 - \beta(\omega_1^2 + \omega_1 \omega_2) + \Gamma_2=0
\end{equation}
\end{subequations}
respectively. Therefore, from \eqref{dyneq12} and \eqref{dyneq22},
\begin{subequations}
\begin{equation}\label{dyneq13app}
\dot{\omega}_1= \frac{1}{2K_1}(u_1 -\Gamma_1 - K_2 \dot{\omega}_2 - \alpha(2\omega_1 + \omega_2))
\end{equation}
\begin{align}\label{dyneq23app}
(2K_3 - \frac{K_2^2}{2K_1})\dot{\omega}_2 &= -\frac{K_2}{2K_1}(u_1-\Gamma_1 - \alpha(2\omega_1 + \omega_2))- \alpha \omega_1 \\ \nonumber
&+ \beta(\omega_1^2 + \omega_1 \omega_2) - \Gamma_2.
\end{align}
\end{subequations}
The reconstruction equations are given by
\begin{equation}\label{reconeq}
\dot{R}_i = R_i {\hat{\omega}}_i, \qquad i =1,2.
\end{equation}
where ${\hat{\omega}}_i \coloneq \begin{pmatrix}
                                   0 & -\omega_i \\
                                   \omega_i & 0
                                 \end{pmatrix}$, for $i =1,2$.

\bibliographystyle{plain}
\bibliography{aps1}
\end{document}